\documentclass[11pt]{article}
\usepackage[a4paper, total={6in, 8in}]{geometry}
\usepackage{graphicx} 
\usepackage{amsmath}
\usepackage{amssymb}
\usepackage{amsthm}
\usepackage{todonotes}
\usepackage{ulem}
\usepackage{authblk}
\usepackage{url}

\newcommand{\idop}{\mathbb{I}}
\newcommand{\ket}[1]{|#1\rangle}
\newcommand{\bra}[1]{\langle#1|}
\newcommand{\ketbra}[1]{|#1\rangle\langle#1|}

\newcommand{\dop}[1]{\mathbf{S}\left(#1\right)}

\newcommand{\cdim}[1]{\mathbb{C}^{#1}}
\newcommand{\lpnorm}[2]{\left\|#2\right\|_{#1}}

\newcommand{\tr}[1]{\text{tr}\left[#1\right]}
\newcommand{\topeigval}[1]{\lambda_{{\rm max}}\left(#1\right)}
\newcommand{\expectation}[1]{\mathbb{E}\left[#1\right]}
\newcommand{\prob}[1]{\mathrm{Pr}\left[#1\right]}

\newcommand{\nn}{\mathbb{N}}
\newcommand{\rr}{\mathbb{R}}

\newcommand{\rvarX}{\mathbf{X}}
\newcommand{\rvarY}{\mathbf{Y}}

\newcommand{\Id}{I}
\DeclareMathOperator{\Range}{Range}

\newtheorem{lemma}{Lemma}

\newtheorem{theorem}{Theorem}
\newtheorem{proposition}{Proposition}
\newtheorem{corollary}{Corollary}

\title{Measurement Geometry for Quantum Random Access Codes: Beyond Nayak Bound and Toward Optimality}
\author[1,2]{Seiseki Akibue\thanks{corresponding author}}
\author[3]{Rudy Raymond}
\author[4]{Suguru Tamaki}
\author[5]{Kosei Teramoto}
\affil[1]{Communication Science Laboratories, NTT, Inc.}
\affil[2]{Sorbonne Université, CNRS, LIP6}
\affil[3]{WPI Bio2Q, Keio University}
\affil[4]{University of Hyogo}
\affil[5]{Mavericks, Inc.}

\begin{document}

\maketitle

\begin{abstract}
Quantum random access codes (QRACs) ask how well \(N\) classical bits can be encoded into \(M\) qubits while allowing any single bit to be recovered. 
Although the Nayak bound remains the standard general upper bound on the decoding probability, numerical evidence suggests a stronger upper bound in the small-qubit regime.

In this work, we formulate the optimal decoding probability in terms of decoding measurements, reformulating QRAC design as a spectral problem for noncommuting measurements. 
Using this formulation, we give an elementary proof of the Nayak bound by simplifying the Chernoff-bound argument. 
Moreover, we refine the argument to obtain upper bounds that improve over Nayak's bound in the entire finite-size regime.

The equality conditions of our bounds justify defining mutually unbiased projector-valued measurements (MUPVMs), a generalization of mutually unbiased bases. 
We show that decoding measurement of any two-qubit QRAC attaining the conjectured bound must form MUPVMs.
We also show that any MUPVM, assisted by one ancillary qubit, yields a QRAC with optimal \(N\)-scaling decoding probability. 
Finally, we propose a new MUPVM-based construction for the \((M+2,M)\)-QRAC family attaining the conjectured bound.
\end{abstract}

\section{Introduction}
Quantum random access codes (QRACs) are a basic primitive for studying how classical information can be compressed into quantum systems while retaining partial accessibility. 
In an \((N,M,p)\)-QRAC, Alice encodes an \(N\)-bit string into an \(M\)-qubit state, and Bob, given an index \(n\in[N]\), attempts to recover the \(n\)-th bit with success probability at least \(p\). 
This task captures a fundamental trade-off between compression rate and decoding probability, and has been studied in quantum information theory~\cite{Nayak99,WBook}, quantum finite automata~\cite{ANTV99,Nayak99}, contextuality~\cite{G02,SBTP09}, foundations of quantum mechanics~\cite{PPKSWZ09}, network coding~\cite{HINRY07}, cryptography~\cite{PB11}, dimension witnesses~\cite{WCD08}, certification of quantum measurements~\cite{TSVBB20,CHT20}, and combinatorial optimization~\cite{FHGIITJKBRM24,TRI23,Kondo2025recursivequantum}. 
Since certification, witnessing, and cryptographic protocols rely on thresholds for quantum advantage to determine their success, optimizing QRACs can enhance the feasibility of experimental implementations.

Despite their importance and long history, the optimal decoding probability of QRACs remains largely unresolved. 
The celebrated general upper bound is the Nayak bound, which follows from limitations on storing classical information in quantum systems~\cite{Nayak99}. 
For an \((N,M,p)\)-QRAC, it implies
\[
        M\geq (1-H_2(p))N
        \quad
        \Leftrightarrow
        \quad
        p\leq H_2^{-1}\left(1-\frac{M}{N}\right),
\]
where \(H_2(p)=-p\log_2 p-(1-p)\log_2(1-p)\) is the binary entropy function and the inverse is taken on \(p\in[1/2,1]\). 
Although even classical RACs can asymptotically attain the Nayak bound~\cite{ANTV99}, the bound is not believed to be tight for finite-dimensional QRACs. 
Numerical evidence and known constructions instead suggest the stronger finite-size behavior
\begin{equation}
\label{eq:num_conj}    
        p\leq \frac12\left(1+\sqrt{\frac{M}{N}}\right),
\end{equation}
particularly in the small-qubit regime. Figure~\ref{fig:nayak-vs-conjecture} compares this conjectured bound with the Nayak bound. 
This conjectured bound was proved for \(M=1\)~\cite{ambainis2009quantumrandomaccesscodes}, and Man\v{c}inska and Storgaard proved it for \(M\leq2\)~\cite{MS22}.

\begin{figure}
    \centering
    \includegraphics[width=0.8\linewidth]{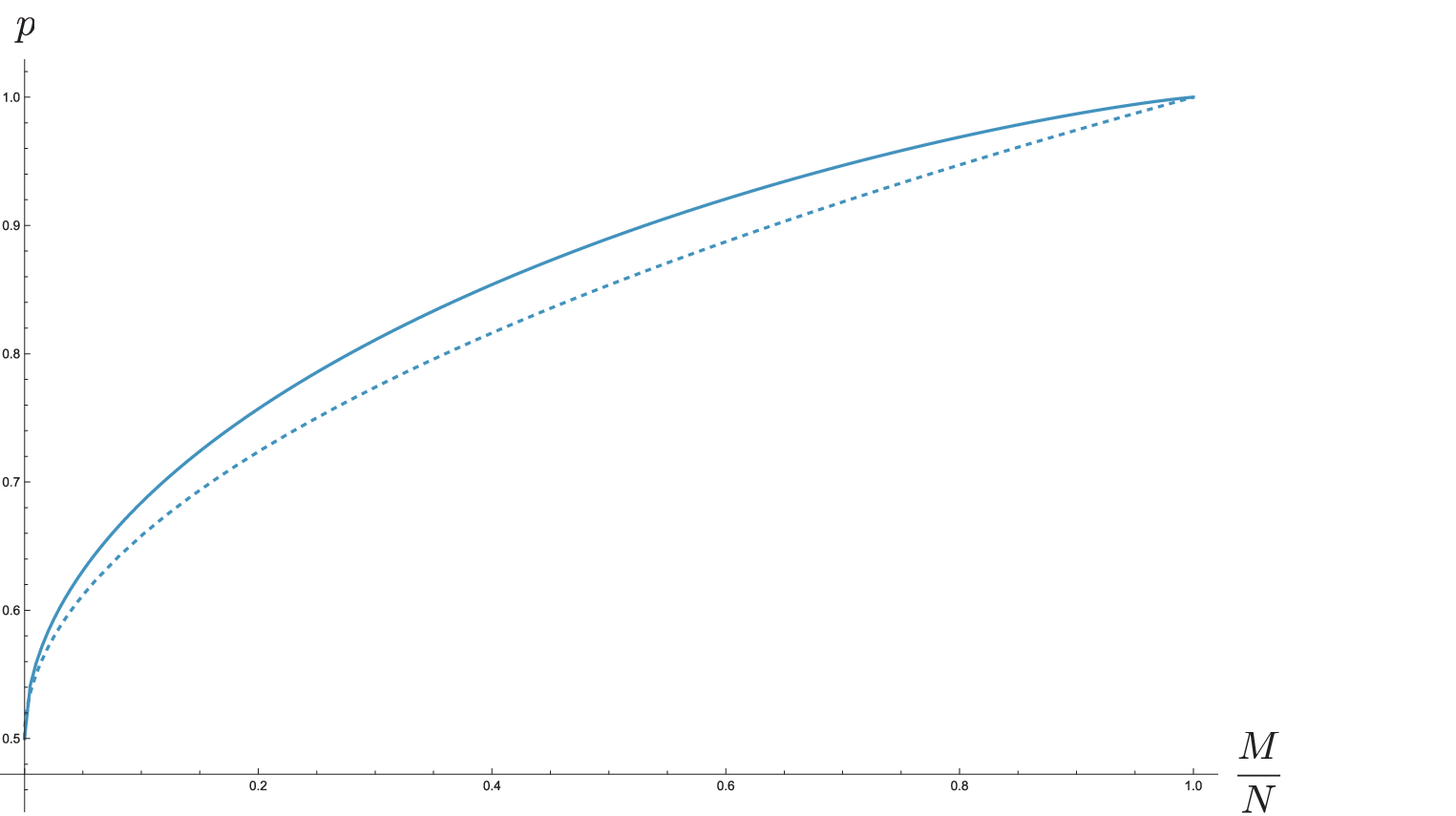}
    \caption{Comparison of upper bounds on the decoding probability $p$. The solid curve represents the Nayak bound $p \le H_{2}^{-1}\left(1-\frac{M}{N}\right)$. The dashed curve represents the conjectured bound $p\leq\frac{1}{2}\left(1+\sqrt{\frac{M}{N}}\right)$.}
    \label{fig:nayak-vs-conjecture}
\end{figure}

In the analysis of finite-size QRACs, the geometry of the Bloch sphere has played a central role~\cite{INRY07,ambainis2009quantumrandomaccesscodes,MS22,Farkas2025simplegeneralbounds,ABMP18}. 
This geometric viewpoint also provides useful intuition for designing encoding states. 
Although substantial effort has been devoted to finding sharper bounds for general QRACs using qudits~\cite{ambainis2009quantumrandomaccesscodes,THMB15,FK19,Vicente_2019,Farkas2025simplegeneralbounds}, the best known upper bounds for binary QRACs, which encode bits in qubits, remain the complementary bounds due to Nayak for \(M>2\) and Man\v{c}inska--Storgaard for \(M\leq2\). 
Moreover, we still have limited structural understanding of how optimal QRACs should be designed in higher-dimensional Hilbert spaces. 
This work addresses both issues by shifting the emphasis from encoding states to decoding measurements.

\subsection{Our Contributions}
The main conceptual novelty of this paper is a measurement-first approach to QRACs: we treat decoding measurements, rather than encoding states, as the primary objects to be designed and analyzed. 
This perspective leads to improvements over the Nayak bound for the entire finite-size regime and reveals geometric constraints that optimal decoders must satisfy.

\begin{enumerate}
    \item \textbf{Optimal decoding probability from decoding measurements.}
    We formulate binary \((N,M,p)\)-QRACs in terms of their decoding observables
    \[
        \{\sigma^{(n)}\}_{n=1}^N .
    \]
    For each classical string \(x\in\{-1,1\}^N\), or $x \in \{\pm\}^N$ for short, the corresponding operator is given by a signed sum
    \[
        Y_x=\sum_{n=1}^N {x_n}\sigma^{(n)} .
    \]
    We show that both the optimal average-case and worst-case decoding probabilities can be expressed through the largest eigenvalues $\lambda_{\max}(Y_x)$ of these signed sums. 
    Thus, QRAC design can be reformulated as a spectral problem for a family of generally noncommuting observables. Details are given in Section~\ref{sec:opt_prob_obs}.

    \item \textbf{Improved Nayak bound via a Chernoff-type argument.}
    The above reformulation naturally suggests applying matrix-concentration ideas to bound the decoding probability. 
    We first recover the Nayak bound by simplifying the matrix Chernoff argument. 
    Remarkably, this gives a simple proof of the Nayak bound that does not rely on strong subadditivity or operator monotonicity, which are used in the original information-theoretic and matrix-concentration proof.
    Moreover, we sharpen this argument by optimizing the moment-generating function used to bound \(\lambda_{\max}(Y_x)\). 
    This yields upper bounds on QRAC decoding probability that are strictly stronger than the Nayak bound in finite regimes. 
    In particular, our bounds reproduce the Man\v{c}inska--Storgaard bound.
    Details are given in Section~\ref{sec:finite_QRAC_bound}.

    \item \textbf{Geometric constraints on optimal decoders.}
    Since our bound matches the conjectured bound in Eq.~\eqref{eq:num_conj} for \(M\le2\), we derive equivalent conditions for the decoding observables to attain it. 
    These conditions imply, in particular, that an \((N,2)\)-QRAC can attain the conjectured bound if and only if $N\in\{2,3,4,6\}$. 
    They also show that optimal QRACs require strong spectral regularity of the signed sums \(Y_x\). 
    We formulate this property by introducing mutually unbiased projector-valued measurements (MUPVMs), a high-rank projective generalization of mutually unbiased bases.

    To demonstrate the significance of MUPVMs in higher-dimensional settings, we prove that any MUPVM can be used, with the assistance of a single-qubit ancilla, as a decoder of a QRAC with decoding probability
    \begin{equation}
    \label{eq:MUPVM_bound}
        p
        \ge
        \frac12\left(1+\frac1{\sqrt N}\right).
    \end{equation}
    Using results from matrix discrepancy, we argue this scaling is optimal in its \(N\)-dependence, especially in the regime \(N\gtrsim 2^M\).

    We also identify a distinguished subclass of MUPVMs: mutually unbiased measurements (MUMs), previously studied in the literature~\cite{FKN23,Kalev_2014}. 
    We show that QRACs based on MUMs attain the bound in Eq.~\eqref{eq:MUPVM_bound}. 
    However, MUMs exist only in the restricted regime \(N\le 2M+1\). 
    Consequently, although MUM-based QRACs are structurally clean, one must go beyond MUMs and consider more general MUPVMs in order to approach the conjectured bound (Eq.~\eqref{eq:num_conj}) or to encode a larger number of classical bits. 
    Details are given in Section~\ref{sec:Opt_decoder_geometry}.

    \item \textbf{QRAC constructions from MUPVMs.}
    Based on MUPVMs and additional heuristic structures, we propose a construction for the \((M+2,M)\) family that attains the conjectured bound (Eq.~\eqref{eq:num_conj}) and thereby adding to the Suzuki's \((M+1,M)\)-QRAC construction. We also improve the known $(5,2)$-QRAC.
    Details are given in Section~\ref{sec:QRAC_construction}.
\end{enumerate}

\section{Notations}
We denote $\nn=\{1,2,3,\cdots\}$ and 
$[N]=\{1,2,\cdots,N\}$.
We consider finite dimensional Hilbert spaces.
We identify every operator with its matrix representation with respect to the computational basis.
$\idop$ represents the the identity matrix.
A positive operator-valued measure (POVM) with a finite outcome set \(\Omega\) is a collection of positive semidefinite operators
$\{M_a\}_{a\in\Omega}$
acting on a Hilbert space \(\mathcal H\), satisfying the normalization condition
\[
    M_a\ge0
    \quad\text{for all }a\in\Omega,
    \qquad
    \sum_{a\in\Omega}M_a=\idop.
\]
$\dop{\mathcal{H}}:=\{\rho\in\mathrm{Hermitian\ operators\ acting\ on\ }\mathcal{H}:\rho\geq0,\tr{\rho}=1\}$ represents the set of density operators.

\section{Optimal decoding probability from decoding measurements}
\label{sec:opt_prob_obs}
Let $\{M_+^{(n)},M_-^{(n)}\}$ be the POVM used to decode the $n$-th bit. We define the corresponding $M$-qubit Hermitian matrix by $\sigma^{(n)}:=2M_+^{(n)}-\idop$. 
We call \(\sigma^{(n)}\) a decoding observable, since it coincides with the physical observable associated with the measurement \(\{M_+^{(n)},M_-^{(n)}\}\) when the measurement is projective. 
We use the same terminology even when the associated measurement is not projective.
Then $M_\pm^{(n)}=\frac{1}{2}(\idop\pm\sigma^{(n)})$ and $-\idop\leq\sigma^{(n)}\leq\idop$ hold for any POVM.

For fixed decoding POVMs, we optimize over all possible encoding states. 
The optimal average-case decoding probability is given by
\begin{eqnarray}
 p_{\rm avg}\left(\{\sigma^{(n)}\}_{n=1}^N\right)&:=&\frac{1}{N2^N}\max_{\{\rho_{x}\}_{x\in \{\pm\}^N}}\sum_{n,x}\tr{\rho_{x}M_{x_n}^{(n)}}\\
 &=&\frac{1}{N2^N}\sum_{x\in \{\pm\}^N}\max_{\rho\in\dop{\cdim{2^M}}}\sum_{n=1}^N\tr{\rho M_{x_n}^{(n)}}\\
&=&\frac{1}{2}\left(1+\frac{1}{N}\expectation{\max_{\rho\in\dop{\cdim{2^M}}}\sum_{n=1}^N\tr{\rho (x_n\sigma^{(n)})}}\right)\\
\label{eq:avgprob}
&=&\frac{1}{2}\left(1+\frac{1}{N}\expectation{\topeigval{\sum_{n=1}^N x_n\sigma^{(n)}}}\right),
 \end{eqnarray}
where the expectation value is taken over iid uniform random variables $x\in \{\pm\}^N$.
Here, the second equality follows because the encoding states \(\rho_x\) can be optimized independently for each \(x\). 
This calculation also shows that an optimal encoding state \(\rho_x\) may be chosen with support contained in the eigenspace of the largest eigenvalue of $Y_x=\sum_{n} x_n\sigma^{(n)}$.

We can similarly formulate the optimal worst-case decoding probability over all possible encodings:
\begin{eqnarray}
 p_{\rm wc}\left(\{\sigma^{(n)}\}_{n=1}^N\right)&:=&\max_{\{\rho_{x}\}_{x\in \{\pm\}^N}}\min_{n,x}\tr{\rho_{x}M_{x_n}^{(n)}}\\
 &=&\min_{x\in \{\pm\}^N}\max_{\rho\in\dop{\cdim{2^M}}}\min_n\tr{\rho M_{x_n}^{(n)}}\\
 \label{eq:worst_prob_cal}
 &=&\min_{x\in \{\pm\}^N}\max_{\rho\in\dop{\cdim{2^M}}}\min_p\sum_{n=1}^Np(n)\tr{\rho M_{x_n}^{(n)}}.
 \end{eqnarray}
 Here \(p\) ranges over probability distributions on \([N]\). 
Since $f(p,\rho)=\sum_{n=1}^Np(n)\tr{\rho M_{x_n}^{(n)}}$ is bilinear in \(p\) and \(\rho\), and since the feasible regions for both variables are compact and convex, von Neumann's minimax theorem gives
 \begin{eqnarray}
Eq.~\eqref{eq:worst_prob_cal} &=&\min_{x\in \{\pm\}^N}\min_p\max_{\rho\in\dop{\cdim{2^M}}}\sum_{n=1}^Np(n)\tr{\rho M_{x_n}^{(n)}}\\
 &=&\frac{1}{2}\left(1+\min_{x\in \{\pm\}^N}\min_p\max_{\rho\in\dop{\cdim{2^M}}}\sum_{n=1}^Np(n)\tr{\rho (x_n\sigma^{(n)})}\right)\\
 \label{eq:worstprob_exact}
&=&\frac{1}{2}\left(1+\min_{x\in \{\pm\}^N}\min_p\topeigval{\sum_{n=1}^Nx_np(n)\sigma^{(n)}}\right)\\
\label{eq:worstprob}
&\leq&\frac{1}{2}\left(1+\frac{1}{N}\min_{x\in \{\pm\}^N}\topeigval{\sum_{n=1}^Nx_n\sigma^{(n)}}\right),
 \end{eqnarray}
where the last inequality follows by choosing \(p\) to be the uniform distribution on \([N]\).

Note that the maximizing state \(\rho\) for fixed \(x\) and \(p\) may be chosen to be pure. 
However, this does not imply that an optimal worst-case encoding can always be chosen pure, because the minimax theorem usually changes the optimal strategy. That is, the optimal encoding may require mixed states.

\section{Improved Nayak bound via a Chernoff-type argument}
\label{sec:finite_QRAC_bound}

We first prove the following proposition, which can be used to prove upper bounds on the worst-case decoding probability $p_{\rm wc}$.
\begin{proposition}
    Let $\rvarY$ be a random Hermitian matrix. Then, for every $t\in\rr$ and every monotonically non-decreasing function $f:\rr\rightarrow[0,\infty)$, we have
    \begin{equation}
        f(t)\prob{\lambda_{\max}(\rvarY)\geq t}\leq \expectation{\tr{f(\rvarY)}}.
    \end{equation}
\end{proposition}
\begin{proof}
    Since \(f\) is non-decreasing, the event \(\{\lambda_{\max}(\rvarY)\ge t\}\) is contained in
    \(\{f(\lambda_{\max}(\rvarY))\ge f(t)\}\). 
    Hence, by Markov's inequality,    \begin{eqnarray}
        \label{eq:Markov2}
        f(t)\prob{\lambda_{\max}(\rvarY)\geq t}&\leq&f(t)\prob{f(\lambda_{\max}(\rvarY))\geq f( t)}\\
        &\leq& \expectation{f(\lambda_{\max}(\rvarY))}.
    \end{eqnarray}
    
    For any Hermitian matrix \(H\), the quantity \(f(\lambda_{\max}(H))\) is one of the eigenvalues of \(f(H)\). 
    Since \(f\) is positive, all eigenvalues of \(f(H)\) are positive, and therefore
    \[
        f(\lambda_{\max}(H))\le \tr{f(H)}.
    \]
    Applying this inequality to \(H=\rvarY\) completes the proof.
\end{proof}

When \(f(x)=e^{\theta x}\) with a positive parameter \(\theta>0\), this proposition reduces to the standard Laplace-transform method, which is a central tool in the derivation of matrix Chernoff bounds~\cite{Tropp12}.
This proposition shows that if $ \expectation{\tr{f(\rvarY)}}<f(t)$, then there exists an event where $\lambda_{\max}(\rvarY)< t.$

We can strengthen this observation when \(f\) is additionally convex. This proposition will be used in the derivation of upper bounds on the average-case decoding probability $p_{\rm avg}$.

\begin{proposition}
    \label{prop:Exp_Markov}
    Let $R\in(0,\infty]$. 
    Let $\rvarY$ be an random Hermitian matrix whose eigenvalues lie in $(-R,R)$. Let $f:(-R,R)\rightarrow[0,\infty)$ be convex and monotonically non-decreasing.
    Then, for every $t\in(-R,R)$,
    \begin{equation}
         \expectation{\tr{f(\rvarY)}}<f(t)\quad \Rightarrow\quad \expectation{\lambda_{\max}(\rvarY)}< t.
    \end{equation}
\end{proposition}
\begin{proof}
    By the same eigenvalue argument as above, for every Hermitian matrix \(H\), $f(\lambda_{\max}(H))\le \tr{f(H)}$.
    Hence,
    \[
        \expectation{f(\lambda_{\max}(\rvarY))}\leq
         \expectation{\tr{f(\rvarY)}}< f(t).
    \]
    Since \(f\) is convex, Jensen's inequality gives
    \[
        f\!\left(\expectation{\lambda_{\max}(\rvarY)}\right)
        \le
        \expectation{f(\lambda_{\max}(\rvarY))}
        < f(t).
    \]
    Since \(f\) is non-decreasing, this inequality
    implies $\expectation{\lambda_{\max}(\rvarY)}< t$.
\end{proof}

In the following proposition, we assume that $f$ is analytic, i.e., $f$ admits a power-series expansion $f(x)=\sum_{k=0}^\infty a_kx^k$ with radius of convergence $R\in(0,\infty]$, so that the series converges for $|x|<R$ and the domain of $f$ is $(-R,R)$. Note that $f(H)$ is well defined and its Taylor expansion converges as $f(H)=\sum_{k=0}^\infty \frac{f^{(k)}(0)}{k!}H^k$ if the eigenvalues of a Hermitian matrix $H$ lie in $(-R,R)$.

\begin{proposition}
    \label{prop:Moment_bound}
    Let $f:(-R,R)\rightarrow\rr$ be an analytic function such that $f^{(2k)}(0)\geq0$ for all $k\in\nn$, and assume that $R>N$.
    Let $\rvarY=\sum_{n=1}^N{x_n}\sigma^{(n)}$ where $x_n\in\{-1,1\}$ are iid uniform random variables and $\sigma^{(n)}$ are Hermitian matrices satisfying $\lpnorm{\infty}{\sigma^{(n)}}\leq1$.
    Then,
    \begin{equation}
        \expectation{\tr{f(\rvarY)}}\leq 2^M\expectation{f\left(\sum_{n=1}^N{x_n}\right)}=\frac{2^M}{2^N}\sum_{k=0}^N
        \begin{pmatrix}
            N\\k
        \end{pmatrix}
        f(N-2k).
    \end{equation}
\end{proposition}
\begin{proof}
    Since the equality can be verified via a straightforward calculation, we show the inequality in the following.
    Since $\lpnorm{\infty}{\rvarY}\leq N$ for every event, the Taylor expansion of $f(\rvarY)$ converges. By using this expansion, we obtain
    \begin{equation}
        \expectation{\tr{f(\rvarY)}}=\tr{\expectation{f(\rvarY)}}=\tr{\expectation{\sum_{k=0}^\infty\frac{f^{(k)}(0)}{k!}\rvarY^k}}.
    \end{equation}
    Since $\sum_{k=0}^\infty\frac{f^{(k)}(0)}{k!}\rvarY^k$ converges for every event, we can interchange the summation with both the expectation and the trace, and proceed as follows:    
    \begin{eqnarray}
        &&\tr{\expectation{\sum_{k=0}^\infty\frac{f^{(k)}(0)}{k!}\rvarY^k}}\\
        &=&\sum_{k=0}^\infty\frac{f^{(k)}(0)}{k!}\tr{\expectation{\rvarY^k}}\\
        \label{eq:moment_cal1}        &=&f(0)\tr{\idop}+\sum_{k=1}^\infty\frac{f^{(k)}(0)}{k!}\tr{\expectation{\sum_{\vec{n}\in[N]^k}\prod_{i=1}^k{x_{n_i}}\sigma^{(n_i)}}}.
    \end{eqnarray}

    Observe that for every $\vec{n}\in[N]^k$,
    \begin{equation}
        \expectation{\prod_{i=1}^k{x_{n_i}}\sigma^{(n_i)}}
        =\prod_{n:|\{i\in[k]:n_{i}=n\}|{\rm\ is\ odd}}\expectation{{x_n}}\prod_{i=1}^k\sigma^{(n_i)}.
    \end{equation}
    By defining
    \begin{equation}
        \mathrm{EVEN}(k):=\{(n_1,\cdots,n_{k})\in[N]^{k}:\forall n\in[N],|\{i\in[k]:n_{i}=n\}|{\rm\ is\ even}\},
    \end{equation}
    we find that
    \begin{equation}
        \expectation{\prod_{i=1}^k{x_{n_i}}\sigma^{(n_i)}}
        =\left\{
        \begin{array}{ll}
             0&{\rm if\ }\vec{n}\notin\mathrm{EVEN}(k)\\
            \prod_{i=1}^k\sigma^{(n_i)}&{\rm if\ }\vec{n}\in\mathrm{EVEN}(k).  
        \end{array}\right.
    \end{equation}
    Since $\mathrm{EVEN}(2k-1)=\emptyset$ for any $k\in\nn$, we proceed as follows

    \begin{eqnarray}
        \label{eq:momment_cal2}
        Eq.~\eqref{eq:moment_cal1}&=&f(0)\tr{\idop}+\sum_{k=1}^\infty\frac{f^{(2k)}(0)}{(2k)!}\tr{\sum_{\vec{n}\in\mathrm{EVEN}(2k)}\prod_{i=1}^{2k}\sigma^{(n_i)}}.
    \end{eqnarray}

    Since $f^{(2k)}(0)\geq0$ and $|\tr{\sigma^{(n_1)}\cdots\sigma^{(n_k)}}|\leq\lpnorm{1}{\sigma^{(n_k)}}\leq2^M$, we obtain
    \begin{eqnarray}
        Eq.~\eqref{eq:momment_cal2}
        &\leq&2^M\left(f(0)+\sum_{k=1}^\infty\frac{f^{(2k)}(0)}{(2k)!}\left|\mathrm{EVEN}(2k)\right|\right).
    \end{eqnarray}    
    On the other hand, when $\sigma^{(n)}=1$ for all $n$ (i.e., in the $1\times1$ matrix case), Eq.~\eqref{eq:momment_cal2} implies
    \begin{equation}
        \expectation{f\left(\sum_{n=1}^N{x_n}\right)}=f(0)+\sum_{k=1}^\infty\frac{f^{(2k)}(0)}{(2k)!}\left|\mathrm{EVEN}(2k)\right|.
    \end{equation}
    This completes the proof.
\end{proof}

When $f(x)=e^{\theta x}$ with a positive parameter $\theta>0$, which is convex, monotonically non-decreasing and analytic, we can derive the Nayak bound as follows:
\begin{theorem}[Simpler derivation of Nayak bound]
    Let $N\geq M$. Any $(N,M)$-QRAC satisfies
    \begin{equation}
            p_{\rm avg} \le H_{2}^{-1}\left(1-\frac{M}{N}\right),
    \end{equation}
    where \(H_2(p)=-p\log_2 p-(1-p)\log_2(1-p)\) is the binary entropy function and the inverse is taken on \(p\in[1/2,1]\). 
\end{theorem}
\begin{proof}
    Since the statement is trivial if $N=M$, we consider the case where $N>M$.
    By setting $f(x)=e^{\theta x}$ with a positive parameter $\theta>0$, Proposition \ref{prop:Moment_bound} implies
    \begin{equation}
        \expectation{\tr{\exp(\theta\rvarY)}}\leq\frac{2^M}{2^N}\sum_{k=0}^N\begin{pmatrix}
            N\\k
        \end{pmatrix}
        \exp(\theta(N-2k))=2^M(\cosh\theta)^N,
    \end{equation}
    where $\rvarY=\sum_{n=1}^N{x_n}\sigma^{(n)}$, $\sigma^{(n)}$ are decoding observables of $(N,M)$-QRAC and $x_n\in\{\pm1\}$ are iid uniform random variables.
    Combining with Proposition \ref{prop:Exp_Markov}, we obtain 
    \begin{eqnarray}
        &&\expectation{\lambda_{\max}(\rvarY)}\nonumber\\
        &\leq&\inf\left\{t\in\rr:\exists\theta>0,e^{-\theta t}\expectation{\tr{\exp(\theta\rvarY)}}<1\right\}\\
        &=&\inf\left\{t\in\rr:\inf_{\theta>0}e^{-\theta t}\expectation{\tr{\exp(\theta\rvarY)}}<1\right\}\\
        &\leq&\inf\left\{t\in\rr:\inf_{\theta>0}e^{-\theta t}2^{M}(\cosh\theta)^N<1\right\}\\
        &=&\inf\left\{t\in\rr:\inf_{\theta>0}\log_2\left(e^{-\theta t}2^{M}(\cosh\theta)^N\right)<0\right\}\\
        &=&\inf\left\{t\in\rr:M-N+\inf_{\theta>0}N\log_2(e^\theta+e^{-\theta})-(\log_2e)\theta t<0\right\}\\
        \label{eq:Nayak_cal}
        &\leq&\inf\left\{t\in(-N,N):M-N-\frac{1}{2}t\log_2\frac{N+t}{N-t}+N\log_2\frac{2N}{\sqrt{N^2-t^2}}<0\right\},
    \end{eqnarray}
    where we used the fact that $\theta=\frac{1}{2}\ln\frac{N+t}{N-t}$ minimizes the infimum when $|t|<N$ to derive the last inequality.
    By observing Eq.~\eqref{eq:avgprob} implies that $\expectation{\lambda_{\max}(\rvarY)}=N(2p_{\rm avg}-1)\in[0,N)$ if $N>M$, we obtain
    \begin{eqnarray}
        &&\frac{M}{N}-1-\left(p_{\rm avg}-\frac{1}{2}\right)\log_2\frac{p_{\rm avg}}{1-p_{\rm avg}}+\frac{1}{2}\log_2\frac{1}{p_{\rm avg}(1-p_{\rm avg})}\geq0\\
        &\Leftrightarrow&H_2(p_{\rm avg})\geq1-\frac{M}{N}.
    \end{eqnarray}
    This completes the proof.
\end{proof}

While the Nayak bound is known to be tight in the asymptotic regime~\cite{ANTV99,Nayak99}, we can derive sharper bounds in the finite-size regime by allowing a broader class of functions $f$ as follows.
\begin{theorem}[Finite-Size Improvement over the Nayak Bound]
\label{thm:finite_QRAC_bound}
    Let $N\geq M$. Any $(N,M)$-QRAC satisfies
    \begin{equation}
            p_{\rm avg} \le \frac{1}{2}\left(1+\frac{1}{N}\inf_{l\in\nn}\left(\frac{2^M}{2^{N+1}}\sum_{k=0}^N
        \begin{pmatrix}
            N\\k
        \end{pmatrix}
        (N-2k)^{2l}\right)^{\frac{1}{2l}}\right).
    \end{equation}
    Moreover, this bound is sharper than the Nayak bound for any finite $N$ and $M$.
\end{theorem}
\begin{proof}
Let 
\begin{equation}
    f(x)=f_{\rm even}(x)+f_{\rm odd}(x)
\end{equation}
be the decomposition of a function $f:(-R,R)\rightarrow\rr$ into even and odd parts, defined as $f_{\rm even}(x)=\frac{1}{2}(f(x)+f(-x))$ and $f_{\rm odd}(x)=\frac{1}{2}(f(x)-f(-x))$, respectively.
We assume that $f_{\rm even}$ and $f_{\rm odd}$ are obtained as limits of even and odd analytical functions $g_l:(-R,R)\rightarrow\rr$ and $h_l:(-R,R)\rightarrow\rr$, respectively. That is,
\begin{equation}
    \forall x\in(-R,R),\lim_{l\rightarrow\infty}g_l(x)=f_{\rm even}(x)\ \wedge\ \lim_{l\rightarrow\infty}h_l(x)=f_{\rm odd}(x).
\end{equation}
Even and odd analytical functions can be represented as $\sum_{k=0}^\infty a_{2k}x^{2k}$ and $\sum_{k=0}^\infty a_{2k+1}x^{2k+1}$, respectively.
We further assume that all coefficients of power series of $g_l$ is non-negative.
Note that this is equivalent to requiring that $\hat{g}_l(t)=g_l(\sqrt{t})=\sum_{k=0}^\infty a_{2k}t^{k}$ be analytic and absolutely monotone as a function of $t\in[0,R^2)$.
By letting $f_l=g_l+h_l$, this implies
\begin{equation}
    \forall k,l\in\nn,f^{(2k)}_l(0)=g_l^{(2k)}(0)\geq0.
\end{equation}
Moreover, since 
\begin{equation}
    \forall x\in(-R,R),g_l^{(2)}(x)=\sum_{k=1}^\infty \frac{g_l^{(2k)}(0)}{(2k-2)!}x^{2k-2}\geq0,
\end{equation}
$g_l$ is convex in $(-R,R)$ for all $l\in\nn$.
Since the limit of convex functions is convex, we find that $f_{\rm even}$ is also convex in $(-R,R)$.
Furthermore, we can show that $f_{\rm even}$ is monotonically non-decreasing in $[0,R)$ as follows:
First, since $g_l(x)=g_l(0)+\sum_{k=1}^\infty\frac{g_l^{(2k)}(0)}{(2k)!}x^{2k}\geq g_l(0)$, $\min_{x\in(-R,R)}g_l(x)=g_l(0)$. Thus, in the limit, we obtain $\min_{x\in(-R,R)}f_{\rm even}(x)=f_{\rm even}(0)$.
Second, for any $0<x<y<R$, there exists $p\in(0,1)$ such that $x=py$. By using the convexity of $f_{\rm even}$, we obtain
\begin{eqnarray}
    f_{\rm even}(y)-f_{\rm even}(x)&=&f_{\rm even}(y)-f_{\rm even}(py+(1-p)0)   \\
    &\geq&f_{\rm even}(y)-\left(pf_{\rm even}(y)+(1-p)f_{\rm even}(0)\right)\\
    &=&(1-p)(f_{\rm even}(y)-f_{\rm even}(0))\geq0.
\end{eqnarray}

Since $\lim_{l\rightarrow\infty}f_l(\lambda)=f(\lambda)$ for every eigenvalue $\lambda$ of $\rvarY$ on every event, we can interchange the limit with both the expectation and the trace. Applying Proposition \ref{prop:Moment_bound}, we obtain
\begin{eqnarray}
    \expectation{\tr{f(\rvarY)}}&=&\lim_{l\rightarrow\infty}\expectation{\tr{f_l(\rvarY)}}\\
    &\leq&\lim_{l\rightarrow\infty}\frac{2^M}{2^N}\sum_{k=0}^N
        \begin{pmatrix}
            N\\k
        \end{pmatrix}
        f_l(N-2k)\\
    &=&\frac{2^M}{2^N}\sum_{k=0}^N
        \begin{pmatrix}
            N\\k
        \end{pmatrix}
        f(N-2k)=\frac{2^M}{2^N}\sum_{k=0}^N
        \begin{pmatrix}
            N\\k
        \end{pmatrix}
        f_{\rm even}(N-2k),
\end{eqnarray}
where we used the fact that $\sum_{k=0}^N
        \begin{pmatrix}
            N\\k
        \end{pmatrix}
        f_{\rm odd}(N-2k)=0$ in the last equality.

If $f$ is further assumed to satisfy $f((-R,R))\subseteq[0,\infty)$ and to be convex and monotonically non-decreasing, then Proposition \ref{prop:Exp_Markov} implies
\begin{equation}
\label{eq:finiteNayak}
    \frac{2^M}{2^N}\sum_{k=0}^N
        \begin{pmatrix}
            N\\k
        \end{pmatrix}
        f_{\rm even}(N-2k)<f(t)\quad \Rightarrow\quad \expectation{\lambda_{\max}(\rvarY)}< t.
\end{equation}
We will show that, in order to obtain sharper bounds on $\expectation{\lambda_{\max}(\rvarY)}$, we can assume $f(0)=0$ and $f_{\rm odd}(x)=f_{\rm even}(x)$ for $x\in[0,R)$ without loss of generality as follows:
Let $t\in(0,R)$ satisfy the condition in Eq.~\eqref{eq:finiteNayak}.
Since 
$2f_{\rm even}(t)=f(t)+f(-t)\geq f(t)$,
we find 
\begin{eqnarray}
    f_{\rm even}(t)\geq\frac{f(t)}{2}>\frac{2^M}{2\cdot2^N}\sum_{k=0}^N
        \begin{pmatrix}
            N\\k
        \end{pmatrix}f_{\rm even}(0)=\frac{2^M}{2}f(0),
\end{eqnarray}
where we used the fact that $\min_{x\in(-R,R)}f_{\rm even}(x)=f_{\rm even}(0)$ in the second inequality.

Since $f_{\rm even}(t)>f(0)$, we can define
\begin{equation}
    \hat{f}(x)=\left\{
    \begin{array}{ll}
         0 &  {\rm if\ }x\in(-R,0)\\
         \frac{f(t)}{f_{\rm even}(t)-f(0)}(f_{\rm even}(x)-f(0)) & {\rm if\ } x\in[0,R).
    \end{array}
    \right.
\end{equation}
We can verify that $\hat{f}$ is convex and monotonically non-decreasing since
it smoothly glues two convex and monotonically non-decreasing functions.
Moreover, we can verify that $\hat{f}_{\rm even}$ is obtained as limits of even analytical functions whose coefficients in the power series are non-negative since $\hat{f}_{\rm even}(x)=\frac{1}{2}\frac{f(t)}{f_{\rm even}(t)-f(0)}(f_{\rm even}(x)-f(0))$.
We can also verify that $\hat{f}_{\rm odd}$ is obtained as limits of odd analytical functions owing to the Stone-Weierstrass theorem.

Then, we can verify that the condition in Eq.~\eqref{eq:finiteNayak} is satisfied by $\hat{f}$ as follows.
\begin{eqnarray}
    &&\frac{2^M}{2^N}\sum_{k=0}^N
        \begin{pmatrix}
            N\\k
        \end{pmatrix}
        \hat{f}_{\rm even}(N-2k)\\
        &=&
        \frac{f(t)}{f_{\rm even}(t)-f(0)}\frac{2^M}{2\cdot 2^N}\sum_{k=0}^N
        \begin{pmatrix}
            N\\k
        \end{pmatrix}
        (f_{\rm even}(N-2k)-f(0))\\
        &<&\frac{f(t)}{f_{\rm even}(t)-f(0)}\frac{f(t)-2^Mf(0)}{2}\\
        &\leq&\frac{f_{\rm even}(t)-2^{M-1}f(0)}{f_{\rm even}(t)-f(0)}f(t)\\
        &\leq&f(t)=\hat{f}(t),
\end{eqnarray}
where we used the fact that $2f_{\rm even}(t)=f(t)+f(-t)\geq f(t)$ in the second inequality.
This proves that we can assume $f(0)=0$ and $f_{\rm odd}(x)=f_{\rm even}(x)$ for $x\in[0,R)$ without loss of generality in order to obtain sharper bounds on $\expectation{\lambda_{\max}(\rvarY)}$ via Eq.~\eqref{eq:finiteNayak}.
Since $f(t)=2f_{\rm even}(t)$ if $0\leq t< R$ under this assumption, we obtain
\begin{equation}
\label{eq:finiteNayak2}
\expectation{\lambda_{\max}(\rvarY)}
        \leq\inf\left\{t\in(0,R):\inf_{ f_{\rm even}}\frac{2^M}{2^N}\sum_{k=0}^N
        \begin{pmatrix}
            N\\k
        \end{pmatrix}
        \frac{f_{\rm even}(N-2k)}{2f_{\rm even}(t)}<1\right\},
\end{equation}
where the inside infimum is taken over non-zero and even polynomial functions whose coefficients are non-negative.
By letting $f_{\rm even}(x)=\sum_lp_lx^{2l}$ with non-negative coefficients $p_l$, we obtain
\begin{equation}
    \sum_{k=0}^N
        \begin{pmatrix}
            N\\k
        \end{pmatrix}
        \frac{f_{\rm even}(N-2k)}{2f_{\rm even}(t)}=\frac{1}{2}
        \frac{\sum_lp_l\sum_{k=0}^N
        \begin{pmatrix}
            N\\k
        \end{pmatrix}(N-2k)^{2l}}{\sum_lp_lt^{2l}}
\end{equation}
Since
\begin{equation}
    \frac{a+pb}{c+pd}=q\frac{a}{c}+(1-q)\frac{b}{d}\geq\min\left\{\frac{a}{c},\frac{b}{d}\right\}
\end{equation}
holds for any positive reals $a,b,c,d,p$,
where $q=(1+\frac{d}{c}p)^{-1}\in(0,1)$, we find
\begin{equation}
    \inf_{ f_{\rm even}}\sum_{k=0}^N
        \begin{pmatrix}
            N\\k
        \end{pmatrix}
        \frac{f_{\rm even}(N-2k)}{2f_{\rm even}(t)}=
    \inf_{l\in\nn}\sum_{k=0}^N
        \begin{pmatrix}
            N\\k
        \end{pmatrix}
        \frac{(N-2k)^{2l}}{2t^{2l}}.
\end{equation}
Hence, Eq.~\eqref{eq:finiteNayak2} implies that
\begin{equation}
    \label{eq:finite_QRAC_bound}
    \expectation{\lambda_{\max}(\rvarY)}\leq\inf_{l\in\nn}\left(\frac{2^M}{2^{N+1}}\sum_{k=0}^N
        \begin{pmatrix}
            N\\k
        \end{pmatrix}
        (N-2k)^{2l}\right)^{\frac{1}{2l}}.
\end{equation}
Since the optimization on $f$ is performed over a broad class of functions that includes the exponential function, the upper bound on \(p_{\rm avg}\) obtained via Eq.~\eqref{eq:avgprob} is sharper than the Nayak bound.
\end{proof}

By varying $l$ and applying Theorem~\ref{thm:finite_QRAC_bound}, we obtain several closed-form upper bounds on the average-case decoding probability $p_{\rm avg}$, as summarized in Table~\ref{tab:finite_QRAC_bounds}. We also plot the upper bounds obtained from Theorem~\ref{thm:finite_QRAC_bound} in Fig.~\ref{fig:finite_QRAC_bounds}.

\begin{table}
    \centering
    \begin{tabular}{cc}
    \hline
        $l$ & upper bound on $p_{\rm avg}$ \\
    \hline
    \hline
       1  &  $\frac{1}{2}\left(1+\sqrt{\frac{2^{M-1}}{N}}\right)$ \\
       2  &  $\frac{1}{2}\left(1+\sqrt{\frac{1}{N}}\left(2^{M-1}\left(3-\frac{2}{N}\right)\right)^{\frac{1}{4}}\right)$ \\
       3 & $\frac{1}{2}\left(1+\sqrt{\frac{1}{N}}\left(2^{M-1}\left(15-\frac{30}{N}+\frac{16}{N^2}\right)\right)^{\frac{1}{6}}\right)$ \\
        4 & $\frac{1}{2}\left(1+\sqrt{\frac{1}{N}}\left(2^{M-1}\left(105-\frac{420}{N}+\frac{588}{N^2}-\frac{272}{N^3}\right)\right)^{\frac{1}{8}}\right)$ \\
        \hline
    \end{tabular}
    \caption{Closed-form upper bounds on $p_{\rm avg}$ obtained from Theorem~\ref{thm:finite_QRAC_bound}. Note that the case \(l=1\) recovers the result of Man\v{c}inska and Storgaard~\cite{MS22}.}
    \label{tab:finite_QRAC_bounds}
\end{table}

\begin{figure}
    \centering
    \includegraphics[width=0.8\linewidth]{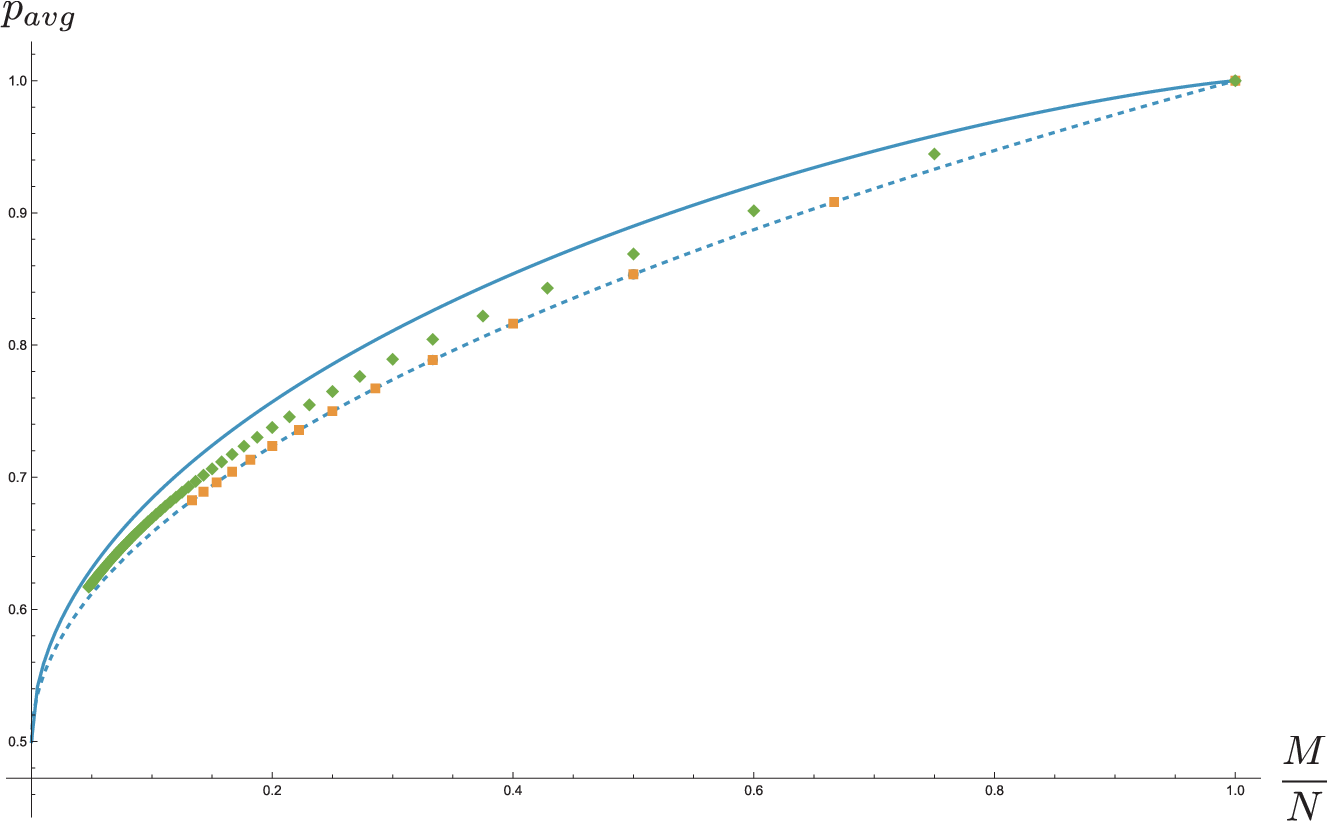}
    \caption{Comparison of upper bounds on the average-case decoding probability $p_{\rm avg}$. The solid curve represents the Nayak bound. The dashed curve represents the conjectured bound.
    The orange points represent the finite-size bounds obtained from Theorem~\ref{thm:finite_QRAC_bound} for \(M=2\), while the green points represent the corresponding new bounds for \(M=3\). 
    For the finite-size bounds, we take \(N\in\{M,M+1,\ldots,4^M-1\}\).}
    \label{fig:finite_QRAC_bounds}
\end{figure}

\section{Geometric constraints on optimal decoders}
\label{sec:Opt_decoder_geometry}
\subsection{Two-qubit QRACs attaining the conjectured bound}
    Since our bound matches the conjectured bound (Eq.~\eqref{eq:num_conj}) for $M\leq2$, we investigate the equality conditions required to attain it.  
    \begin{theorem}
        \label{thm:2qubit_QRAC}
        An $(N,2)$-QRAC attains the average-case decoding probability $p_{\rm avg}=\frac{1}{2}\left(1+\sqrt{\frac{2}{N}}\right)$ if and only if it attains the worst-case decoding probability $p_{\rm wc}=\frac{1}{2}\left(1+\sqrt{\frac{2}{N}}\right)$, and this holds if and only if its two-qubit decoding observables $\{\sigma^{(n)}\}_{n=1}^N$ satisfy the following conditions:
        \begin{enumerate}
            \item $\forall n,\ (\sigma^{(n)})^2=\idop\ \wedge\ \tr{\sigma^{(n)}}=0$,
            \item $\forall n,m,\ n\neq m\Rightarrow\tr{\sigma^{(n)}\sigma^{(m)}}=0$,
            \item $\forall n,\ \sum_m\sigma^{(m)}\sigma^{(n)}\sigma^{(m)}=2\sigma^{(n)}$,
            \item $\sum_{\pi\in S_3}\sigma^{(\pi(k))}\sigma^{(\pi(l))}\sigma^{(\pi(m))}=0$ holds for any distinct $k,l,m$.
        \end{enumerate}
    \end{theorem}
    \begin{proof}
    Since our bound (for $l=1$) implies $p_{\rm wc}\leq p_{\rm avg}\leq\frac{1}{2}\left(1+\sqrt{\frac{2}{N}}\right)$, $p_{\rm wc}=\frac{1}{2}\left(1+\sqrt{\frac{2}{N}}\right)$ implies $p_{\rm avg}=\frac{1}{2}\left(1+\sqrt{\frac{2}{N}}\right)$. We show the opposite direction as follows.
    
    Let $\rvarY=\sum_{n=1}^Nx_n\sigma^{(n)}$ and $x_n\in\{\pm1\}$ are sampled from iid distribution. Then, by assuming $p_{\rm avg}=\frac{1}{2}\left(1+\sqrt{\frac{2}{N}}\right)$, Eq.~\eqref{eq:avgprob} implies
    \begin{equation}
        \label{eq:tight_QRAC_cal1}
        \expectation{\topeigval{\rvarY}}=\sqrt{2N}.
    \end{equation}        
    Define $f:\rr\rightarrow[0,\infty)$ as
        \begin{equation}
            f(x)=\left\{
    \begin{array}{ll}
         0 &  {\rm if\ }x<0\\
         x^2 & {\rm if\ } x\geq0,
    \end{array}
    \right.
\end{equation}
    By using the same argument in the proof of Proposition~\ref{prop:Exp_Markov}, we find
    \begin{eqnarray}
            \label{eq:tight_QRAC_cal2}
            \expectation{\topeigval{\rvarY}}^2\leq\expectation{\topeigval{\rvarY}^2}\leq\expectation{\tr{f(\rvarY)}}.
    \end{eqnarray}
    Proposition~\ref{prop:Moment_bound} implies that
    \begin{equation}
        \label{eq:tight_QRAC_cal3}
        \expectation{\tr{f(\rvarY)}}\leq\frac{4}{2^N}\sum_{k=0}^N
        \begin{pmatrix}
            N\\k
        \end{pmatrix}f(N-2k)=\frac{2}{2^N}\sum_{k=0}^N
        \begin{pmatrix}
            N\\k
        \end{pmatrix}(N-2k)^2=2N.
    \end{equation}
    Eqs.~\eqref{eq:tight_QRAC_cal1}\eqref{eq:tight_QRAC_cal2}\eqref{eq:tight_QRAC_cal3} imply
    \begin{eqnarray}
            \label{eq:tight_QRAC_cal4}
            \expectation{\topeigval{\rvarY}}^2=\expectation{\topeigval{\rvarY}^2}=\expectation{\tr{f(\rvarY)}}=2N.
    \end{eqnarray}
    By Jensen's inequality, the first equation implies $\topeigval{\rvarY}=\sqrt{2N}$ for all events. This shows $p_{\rm wc}=\frac{1}{2}\left(1+\sqrt{\frac{2}{N}}\right)$.
    Moreover, the second equation implies the eigenvalues of $\rvarY$ must be $(\sqrt{2N},0,0,-\sqrt{2N})$ for all events.

    Observe that $\lambda(\rvarY)=(\sqrt{2N},0,0,-\sqrt{2N})$ if and only if
    \begin{equation}
        \label{eq:cond_optimalQRAC_Y}
         \tr{\rvarY}=0,\quad\tr{\rvarY^2}=4N,\quad \rvarY^3=2N\rvarY.
    \end{equation}
    We derive the equivalent conditions for Eq.~\eqref{eq:cond_optimalQRAC_Y} in terms of decoding observables $\sigma^{(n)}$.

Since any polynomial \(p(x)\) admits a unique linear form in the monomials $\left(\prod_{i\in S}x_i\right)_{S\subseteq[N]}$ on the Boolean cube $x\in\{\pm 1\}^N$:
\[
p(x)=\sum_{S\subseteq[N]}\widehat p(S)\prod_{i\in S}x_i,
\] 
$\forall x\in\{\pm1\}^N,p(x)=0$ holds if and only if all Fourier coefficients vanish:
\[
\widehat p(S)=0\quad\text{for every }S\subseteq[N].
\]
By calculating the Fourier coefficients, Eq.~\eqref{eq:cond_optimalQRAC_Y} is equivalent to
\begin{eqnarray}
    \forall n,\tr{\sigma^{(n)}}=0,\\
    \sum_{n=1}^N\tr{(\sigma^{(n)})^2}=4N,\quad\forall n\neq m,\tr{\sigma^{(n)}\sigma^{(m)}}=0,\\
    \forall n,\ \sum_m\sigma^{(m)}\sigma^{(n)}\sigma^{(m)}=2\sigma^{(n)},\\
    \sum_{\pi\in S_3}\sigma^{(\pi(k))}\sigma^{(\pi(l))}\sigma^{(\pi(m))}=0 {\rm\ holds\  for\  any\  distinct\ } k,l,m.
\end{eqnarray}
This completes the proof.
\end{proof}

By using Theorem~\ref{thm:2qubit_QRAC}, we can show that the finite-size QRAC bound in Theorem~\ref{thm:finite_QRAC_bound} is not tight in general. 
Indeed, the first three conditions in Theorem~\ref{thm:2qubit_QRAC} cannot be simultaneously satisfied when \(N>6\).

Let $\mathcal{E}(X)=\sum_{n=1}^N\sigma^{(n)}X\sigma^{(n)}$. 
Let $\{\hat{\sigma}^{(n)}\}_{n=0}^{15}$ be an orthonormal basis of  two-qubit Hermitian matrices such that $\hat{\sigma}^{(0)}=\frac{1}{2}\idop$ and $\hat{\sigma}^{(n)}=\frac{1}{2}\sigma^{(n)}$ for $n\in[N]$.
We define the matrix representation of \(\mathcal{E}\) in this basis by $\hat{\mathcal{E}}_{mn}:=\tr{\hat{\sigma}^{(m)}\mathcal{E}(\hat{\sigma}^{(n)})}$.
Since \(\mathcal{E}\) is self-adjoint with respect to the Hilbert--Schmidt inner product, \(\hat{\mathcal{E}}\) is a \(16\times16\) real symmetric matrix and therefore has \(16\) real eigenvalues.

From the first and third conditions, we have
\begin{equation}
    \mathcal{E}(\sigma^{(n)})=2\sigma^{(n)},\quad\mathcal{E}(\idop)=N\idop.
\end{equation}
Thus, the eigenvalues of \(\hat{\mathcal{E}}\) can be written as $(N,2,\cdots,2,\lambda_1,\cdots,\lambda_{15-N})$.

We next compute the trace of \(\hat{\mathcal{E}}\). We have
\begin{eqnarray}
    \tr{\hat{\mathcal{E}}}&=&\sum_{m=0}^{15}\tr{\hat{\sigma}^{(m)}\mathcal{E}(\hat{\sigma}^{(m)})}\\
    &=&\sum_{n=1}^N\sum_{m=0}^{15}\tr{\hat{\sigma}^{(m)}\sigma^{(n)}\hat{\sigma}^{(m)}\sigma^{(n)}}\\
    &=&\sum_{n=1}^N\sum_{m=0}^{15}\bra{\hat{\sigma}^{(m)}}\left((\sigma^{(n)})^T\otimes\sigma^{(n)}\right)\ket{\hat{\sigma}^{(m)}}\\
    &=&\sum_{n=1}^N\tr{(\sigma^{(n)})^T\otimes\sigma^{(n)}}\\
    \label{eq:2Q_QRAC_impossible1}
    &=&\sum_{n=1}^N\left|\tr{\sigma^{(n)}}\right|^2=0.
\end{eqnarray}
Here we used that $\{\ket{\hat{\sigma}^{(m)}}:=\sum_{i=0}^3\ket{i}\otimes(\hat{\sigma}^{(m)}\ket{i})\}_{m=0}^{15}$ forms an orthonormal basis in the fourth equality. Since $\tr{\hat{\mathcal{E}}}=N+2N+\sum_{m=1}^{15-N}\lambda_m=0$, Eq.~\eqref{eq:2Q_QRAC_impossible1} implies \(N<15\).

Similarly, we obtain
\begin{eqnarray}
    \label{eq:2Q_QRAC_impossible2}
    \tr{\hat{\mathcal{E}}^2}&=&\sum_{n,m=1}^N\left|\tr{\sigma^{(n)}\sigma^{(m)}}\right|^2=16N,
\end{eqnarray}
where the last equality follows from the first and second conditions. 
Combining Eqs.~\eqref{eq:2Q_QRAC_impossible1} and~\eqref{eq:2Q_QRAC_impossible2}, we obtain
\begin{equation}
    N+2N+\sum_{m=1}^{15-N}\lambda_m=0,\quad N^2+4N+\sum_{m=1}^{15-N}\lambda_m^2=16N.
\end{equation}
Equivalently,
\[
    \sum_{m=1}^{15-N}\lambda_m=-3N,
    \qquad
    \sum_{m=1}^{15-N}\lambda_m^2=N(12-N).
\]
Since such real numbers \(\lambda_m\) can exist only if the hyperplane intersects the corresponding hypersphere, by Cauchy-Schwarz, $\left(\sum_{m=1}^{15-N}\lambda_m\right)^2 \le (15-N)\cdot \sum_{m=1}^{15-N}\lambda_m^2$ holds, and gives
\begin{eqnarray}
    \frac{9N}{15-N}\leq 12-N\ \Leftrightarrow\  (N-30)(N-6)\geq0\ \Leftrightarrow\ N\leq 6,
\end{eqnarray}
where we used $N<15$.
Therefore, the first three conditions in Theorem~\ref{thm:2qubit_QRAC} cannot hold for \(N>6\). 
In particular, the finite-size QRAC bound in Theorem~\ref{thm:finite_QRAC_bound} cannot be tight for two-qubit QRACs with \(N>6\).
Additionally, for \(N=5\), we can show non-tightness by a different argument based on a more refined geometric analysis.

\subsection{Mutually unbiased projector-valued measurements}
In the previous section, we have shown that the two-qubit QRACs' decoding observables $\sigma^{(n)}$ are orthogonal traceless involutory matrices with respect to the Hilbert Schmidt inner product if it attains the bounds shown in Theorem~\ref{thm:finite_QRAC_bound}. These conditions are equivalent to requiring that the corresponding POVMs $\{M_+^{(n)},M_-^{(n)}\}$ are projective-valued measures (PVMs) and satisfy
\begin{eqnarray}
    \tr{M_s^{(n)}M_t^{(m)}}&=&\frac{1}{4}\tr{(\idop+s\sigma^{(n)})(\idop+t\sigma^{(m)})}=\frac{D}{d^2}(1-\delta_{nm})+\frac{D}{d}\delta_{st}\delta_{nm}
\end{eqnarray}
for all $n,m$ and for all independent choices of measurement outcomes $s,t\in\{\pm\}$, where $D=4$ and $d=2$.
We can extend this mutually unbiased condition to arbitrary sets of \(D\)-dimensional PVMs with $d$ outcomes, each of whose elements is a projector of equal rank \(D/d\). 
Mutually unbiased bases~\cite{IDIvonovic_1981,WOOTTERS1989363}, or MUBs, are recovered as the special case \(D=d\).

For the analysis of binary QRACs, we only consider binary output measurements. Under this circumstance, the MUPVM condition is equivalent to requiring that the decoding observables \(\sigma^{(n)}\) are traceless and involutory, and satisfy the orthogonality relation
\begin{equation}
\label{eq:MUPVMu}
    \tr{\sigma^{(n)}\sigma^{(m)}}=2^M\delta_{nm}.
\end{equation}
Since the Pauli group satisfies these relations, we can verify that MUPVMs exist with $N=4^M-1$.

To demonstrate the significance of MUPVMs not only in two qubits but also in higher-dimensional settings, we prove the following proposition.
\begin{proposition}
    \label{prop:MUPVM_QRAC}
    For any $M$-qubit MUPVMs $\{M_\pm^{(n)}\}_{n=1}^N$, we can construct an $(N,M+1)$-QRAC whose success probability satisfies
    \begin{equation}
        p_{\rm wc}\geq\frac{1}{2}\left(1+\frac{1}{\sqrt{N}}\right),
    \end{equation}
    and its decoding measurement can be constructed by the quantum circuit depicted in Fig.~\ref{fig:MUPVM_decoder}.
\end{proposition}
\begin{figure}
    \centering
    \includegraphics[width=0.5\linewidth]{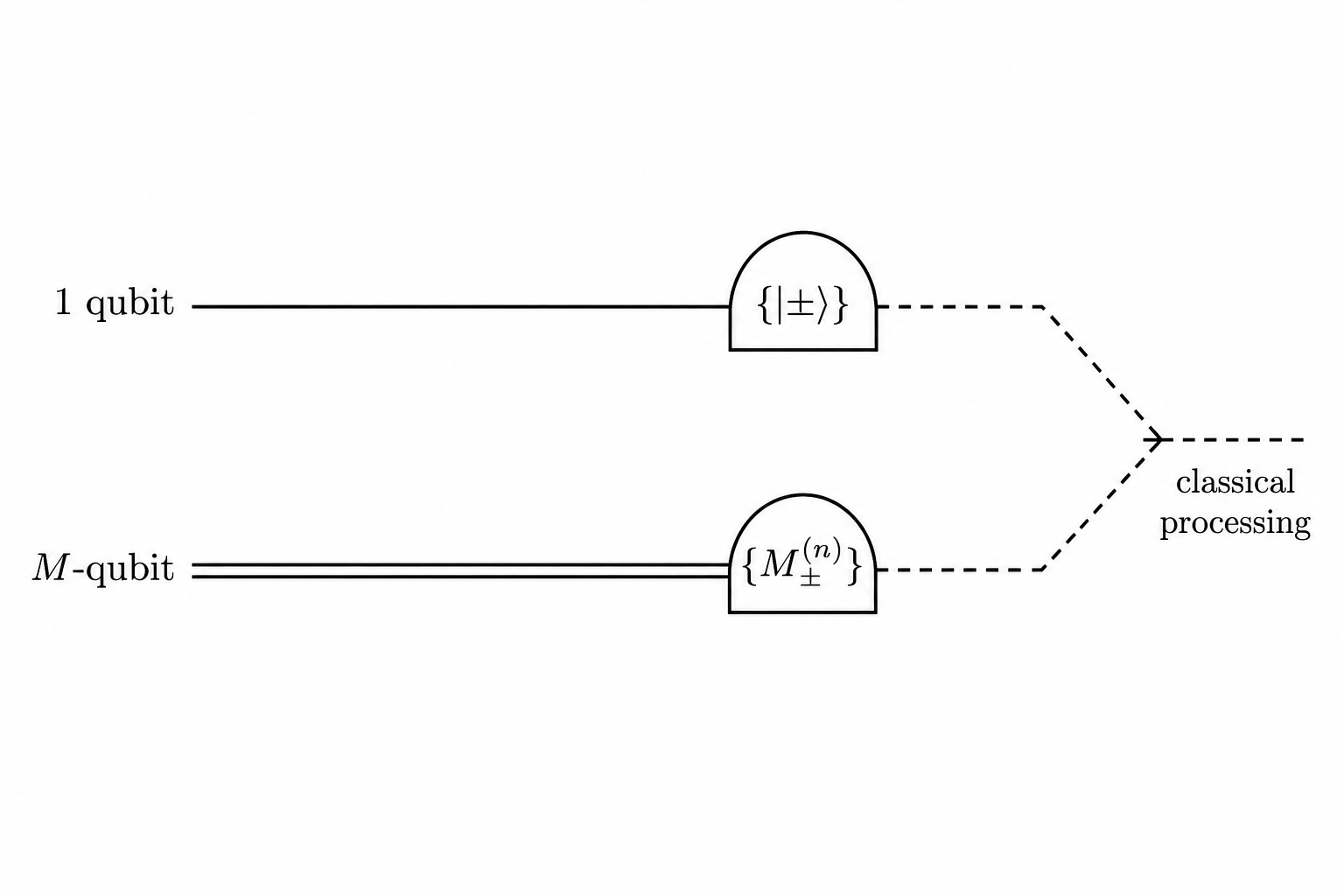}
    \caption{Quantum circuit for the decoding measurement in Proposition~\ref{prop:MUPVM_QRAC}. 
    The first qubit is measured in the computational basis with outcome \(\{\pm\}\), and the remaining \(M\) qubits are measured by the MUPVM with outcome \(\{\pm\}\).
    The classical post-processing outputs \(+\) if the two measurement outcomes coincide, and outputs \(-\) otherwise.}
    \label{fig:MUPVM_decoder}
\end{figure}
Although we consider a slightly different quantity, our proof is based on the argument of Ref.~\cite{DJR22}.
\begin{proof}
    Let $\{\sigma^{(n)}\}_{n=1}^N$ be the decoding observables associated with $\{M_\pm^{(n)}\}_{n=1}^N$.
    First, observe that
    \begin{equation}
        \lpnorm{2}{\sum_{n=1}^Nx_np(n)\sigma^{(n)}}^2=2^M\sum_{n=1}^Np(n)^2\geq \frac{2^M}{N}
    \end{equation}
    for any $x_n\in\{\pm\}$ and probability distribution $p$.
    Since $\sqrt{d}\lpnorm{\infty}{M}\geq\lpnorm{2}{M}$, we obtain
    \begin{equation}
        \lpnorm{\infty}{\sum_{n=1}^Nx_np(n)\sigma^{(n)}}\geq \frac{1}{\sqrt{N}}.
    \end{equation}
    Let $\tilde{\sigma}^{(n)}$ be an $(M+1)$-qubit operator corresponding to the decoding measurement depicted in Fig.~\ref{fig:MUPVM_decoder}. Thus, it can be represented as
    \begin{eqnarray}
        \tilde{\sigma}^{(n)}&=&\begin{pmatrix}
            1&0\\
            0&0
        \end{pmatrix}\otimes M_+^{(n)}+
        \begin{pmatrix}
            0&0\\
            0&1
        \end{pmatrix}\otimes M_-^{(n)}-
        \begin{pmatrix}
            1&0\\
            0&0
        \end{pmatrix}\otimes M_-^{(n)}-
        \begin{pmatrix}
            0&0\\
            0&1
        \end{pmatrix}\otimes M_+^{(n)}
        \\
        &=&\begin{pmatrix}
            1&0\\
            0&-1
        \end{pmatrix}\otimes \sigma^{(n)}.
    \end{eqnarray}
    Since $\lambda(\sum_nr_n\tilde{\sigma}^{(n)})=\{\pm1\}\times\lambda(\sum_nr_n\tilde{\sigma}^{(n)})$ for any $r_n\in\rr$, we find
    \begin{equation}
        \lambda_{\max}\left(\sum_{n=1}^Nx_np(n)\tilde{\sigma}^{(n)}\right)
        =\lpnorm{\infty}{\sum_{n=1}^Nx_np(n)\sigma^{(n)}}.
    \end{equation}
    Combining with Eq.~\eqref{eq:worstprob_exact}, this completes the proof.
\end{proof}
By using Pauli group, we can construct an $(N,\lceil\frac{1}{2}\log_2(N+1)\rceil+1)$-QRAC with the success probability $p_{\rm wc}\geq\frac{1}{2}\left(1+\frac{1}{\sqrt{N}}\right)$.

Using results from matrix discrepancy, we argue this scaling is optimal in its $N$-dependence.
$\min_{x\in \{\pm\}^N}\lpnorm{\infty}{\sum_{n=1}^Nx_n\sigma^{(n)}}$, which provides an upper bound on $p_{\rm wc}$, is the main quantity investigated in the topic of the matrix discrepancy~\cite{HRS22,DJR22,BJM23}.
Indeed, if the conjecture on the matrix Spencer is correct, we obtain
\begin{equation}
    \min_{x\in \{\pm\}^N}\lpnorm{\infty}{\sum_{n=1}^Nx_n\sigma^{(n)}}= O(\sqrt{N\max\{1,M-\log_2 N\}}).
\end{equation}
This implies that there exists a constant $C$ such that for sufficiently large $N$ and $M$,
\begin{equation}
    \label{eq:discrepancy}
    p_{\rm wc}\leq\frac{1}{2}\left(1+C\left(\sqrt{\frac{\max\{1,M-\log_2 N\}}{N}}\right)\right)
\end{equation}
holds. For any positive number $\varepsilon>0$, in the regime $N\geq 2^{(1+\varepsilon)M}$ with large $M$,
the matrix Spencer conjecture has recently been proven~\cite{BJM23}. Thus, in such regime, the worst-case decoding probability given in Proposition \ref{prop:MUPVM_QRAC} has optimal scaling in its $N$-dependence.

\subsection{Relationship between MUPVMs and MUMs}
Farkas et al. have defined the mutually unbiased measurements (MUMs)~\cite{FKN23} as follows: two $D$-dimensional PVMs $\{P_a\}_{a=1}^{d}$ and $\{Q_b\}_{b=1}^{d}$ are MUMs if
\begin{equation}
  \forall a,\forall b,  P_a=d P_aQ_b P_a,\quad Q_b=dQ_bP_aQ_b.
\end{equation}
We can verify that MUMs are MUPVMs since
\begin{eqnarray}
    \forall a,b,\ \tr{P_a}&=&d\tr{P_aQ_b}=\tr{Q_b}=\frac{D}{d},\\
    \forall a,b,\ \tr{P_aQ_b}&=&\tr{P_aQ_bP_a}=\frac{1}{d}\tr{P_a}=\frac{D}{d^2}.
\end{eqnarray}
In the case of $d=2$, MUPVMs are MUMs if and only if the corresponding observable $\{\sigma^{(n)}\}_{n=1}^N$ is an anti-commutative orthonormal set of traceless and involutory matrices, i.e., it satisfies
 \begin{equation}
\forall n\neq m, \{\sigma^{(n)},\sigma^{(m)}\}=0\wedge\tr{\sigma^{(n)}\sigma^{(m)}}=2^M\delta_{n,m}\wedge
\tr{\sigma^{(n)}}=0\wedge(\sigma^{(n)})^2=\idop.
\end{equation}
 
It is known that such a set exists only if $N\leq 2M+1$ and the set of Weyl-Brauer operators is an example of a set with $N= 2M+1$.
We can verify that for any real vector $r\in\rr^{N}$, $M=\sum_{n=1}^Nr_n\sigma^{(n)}$ satisfies
\begin{equation}
\label{eq:antiCMUPVM}
 \tr{M}=0\ \wedge\  M^2=\lpnorm{2}{r}^2\idop.
\end{equation}
This implies that the set of eigenvalues of $M$ is $\{\pm\lpnorm{2}{r}\}$.
Thus, we obtain the following proposition.
\begin{proposition}
    For any $(N,M)$-QRAC using an MUM, the success probability is given by
    \begin{equation}
        p_{\rm wc}=p_{\rm avg}=\frac{1}{2}\left(1+\frac{1}{\sqrt{N}}\right).
    \end{equation}
\end{proposition}
\begin{proof}
    Let $\{\sigma^{(n)}\}_{n=1}^N$ be the decoding observable of a QRAC.
    Since $\topeigval{\sum_{n=1}^Nr_n\sigma^{(n)}}=\lpnorm{2}{r}$ for any real vector $r\in\rr^{N}$, Eq.~\eqref{eq:avgprob} implies
    \begin{eqnarray}
        p_{\rm avg}=\frac{1}{2}\left(1+\frac{1}{N}\expectation{\lpnorm{2}x{}}\right)=\frac{1}{2}\left(1+\frac{1}{\sqrt{N}}\right).
    \end{eqnarray}
    Similarly, by using Eq.~\eqref{eq:worstprob_exact}, we obtain
    \begin{eqnarray}
        p_{\rm wc}&=&\frac{1}{2}\left(1+\min_p\lpnorm{2}{p}\right)=\frac{1}{2}\left(1+\frac{1}{\sqrt{N}}\right).
    \end{eqnarray}
\end{proof}

\section{QRAC constructions from MUPVMs}
\label{sec:QRAC_construction}

As shown in the previous section, the equality conditions of our QRAC bounds naturally lead to the notion of MUPVMs. 
In this section, we present explicit QRAC constructions based on MUPVMs. 
As Theorem~\ref{thm:2qubit_QRAC} shows, however, MUPVMs do not capture all constraints required for optimality. 
We therefore construct QRACs whose decoding observables satisfy several additional constraints, motivated by those appearing in existing QRAC constructions.

\subsection{New construction for \((5,2)\)-QRAC}
We have shown that an \((N,2)\)-QRAC cannot attain the conjectured bound (Eq.~\eqref{eq:num_conj}) when \(N\notin\{2,3,4,6\}\). 
On the other hand, the bound is known to be attainable for \(N\in\{2,3,4,6\}\)~\cite{ANTV99,imamichi2018constructions}. 
In this section, we demonstrate the usefulness of MUPVMs in regimes where an \((N,2)\)-QRAC cannot attain the conjectured bound.

We first construct a $(5,2)$-QRAC with decoding probability $\approx 0.8077$. It is derived from Mutually Unbiased Bases (MUBs).
In the construction, we use Pauli observables $I, X, Y, Z$ or sparse sums of Pauli observables as decoding observables $\sigma^{(n)}$ to guarantee their projectivity and orthogonality.

\begin{proposition}
\label{prop:52}
There exists a $(5,2,1/2(1+\sqrt{5 + 2\sqrt{5}}/5))$-QRAC with decoding observables
\begin{equation}\label{eq:52-obs}
  \sigma^{(1)} = ZI,\quad
  \sigma^{(2)} = XI,\quad
  \sigma^{(3)} = -IY,\quad
  \sigma^{(4)} = -ZX,\quad
  \sigma^{(5)} = -XZ,
\end{equation}
which are orthogonal, traceless, Hermitian involutions, i.e., the associated POVMs are MUPVMs.
Moreover, the decoding probability cannot be improved as long as these observables are used.
\end{proposition}
\begin{proof}
The first equality follows from the following claim.\\
\noindent\textbf{Claim.}~
For every $x \in \{\pm\}^5$, $\lambda_{\max}(Y_x) = \sqrt{5 + 2\sqrt{5}}$.    
\\
\noindent To show the claim, let us expand
\[
  Y_x^2 = \sum_k \left(\sigma^{(k)}\right)^2 + \sum_{k<l}{x_k x_l}\{\sigma^{(k)}, \sigma^{(l)}\}
  = 5\,I + S_x,
\] where 
\[
  S_x := \sum_{k<l}{x_k x_l}\{\sigma^{(k)},\sigma^{(l)}\}.
\]
Direct computation of the ten anticommutators
of Eq.~\eqref{eq:52-obs} gives five vanishing pairs
and five non-vanishing pairs:
\[
  \{\sigma^{(1)},\sigma^{(3)}\} = -2\,ZY,\quad
  \{\sigma^{(1)},\sigma^{(4)}\} = -2\,IX,\quad
  \{\sigma^{(2)},\sigma^{(3)}\} = -2\,XY,
\]
\[
  \{\sigma^{(2)},\sigma^{(5)}\} = -2\,IZ,\quad
  \{\sigma^{(4)},\sigma^{(5)}\} = 2\,YY.
\]
The five non-vanishing anticommutators are, up to a common factor
of $2$, the Pauli strings $\{ZY, IX, XY, IZ, YY\}$.  These are
mutually anticommuting (a direct check on each of the ten pairs)
and each squares to $I$.

Hence $S_x = \sum_{i=1}^{5} \varepsilon_i^{(x)} M_i$ with $\varepsilon_i^{(x)}\in\{\pm\}$ and
$\{M_i\}_{i=1}^5 = \{ 2\,ZY, 2\,IX, 2\,XY, 2\,IZ, 2\,YY\}$
satisfying $\{M_i, M_j\} = 0$ for $i \ne j$ and $M_i^2 = 4I$, so
\[
  S_x^2 = \sum_i M_i^2 = 5 \cdot 4  I = 20\,I,
\]
independently of $x$.  Since $S_x$ is Hermitian, traceless, and
satisfies $S_x^2 = 20\,I$ in dimension $4$, its spectrum is
$\{+2\sqrt{5}, +2\sqrt{5}, -2\sqrt{5}, -2\sqrt{5}\}$.
Therefore $Y_x^2 = 5\,I + S_x$ has spectrum
$\{5+2\sqrt{5}, 5+2\sqrt{5}, 5-2\sqrt{5}, 5-2\sqrt{5}\}$ and
$\lambda_{\max}(Y_x) = \sqrt{5+2\sqrt{5}}$.
Eq.~\eqref{eq:avgprob} implies that $p_{\rm avg}=\frac{1}{2}\left(1+\frac{\sqrt{5+2\sqrt{5}}}{5}\right)$.

The bound $p_{\rm wc}\geq\frac{1}{2}\left(1+\frac{\sqrt{5+2\sqrt{5}}}{5}\right)$ can be verified by a straightforward calculation, taking the encoding states to be eigenvectors of $Y_x$ corresponding to its largest eigenvalues.
\end{proof}

Using a similar technique, we construct additional QRACs, including an improved \((5,2)\)-QRAC, as shown in Appendix~\ref{appendix:QRAC_constructions}.

Before proceeding to the next section, we mention that the encoding states $\{\ket{\psi_x}\}_{x \in \{\pm\}^N}$ and decoding observables $\{\sigma^{(n)}\}_{n=1}^N$ given in Proposition~\ref{prop:52} satisfy the following three conditions with a non-negative parameter $\mu$: \textbf{(TF)}, \textbf{(UD)}, and \textbf{(EC)}.
\begin{itemize}
\item[\textbf{(TF)}] Tight Frame:
  $\displaystyle\frac{2^M}{2^N}\sum_{x} \ket{\psi_x}\!\bra{\psi_x} = \idop$.
\item[\textbf{(UD)}] Uniform Decode:
  $\bra{\psi_x}  \sigma^{(n)} \ket{\psi_x} = \sqrt{\mu}\,{x_n}$
  for all $x, n$.
\item[\textbf{(EC)}] Eigenvalue Channel:
  $\Phi(\sigma^{(n)}) = \mu\, \sigma^{(n)}$, where
  $\Phi(A) := \frac{2^M}{2^N}\sum_{x}
  \bra{\psi_x}A\ket{\psi_x}\ket{\psi_x}\!\bra{\psi_x}$.
\end{itemize}

We can also derive the relationship between encoding states and the observables as follows.
\begin{lemma}
\label{lem:frame-povm}
Under \textup{(TF)}, \textup{(UD)}, \textup{(EC)} at value $\mu > 0$, the encoding states $\{\ket{\psi_x}\}_{x \in \{\pm\}^N}$ and decoding observables $\{\sigma^{(n)}\}_{n=1}^N$ satisfy
\[
  \sigma^{(n)} = \frac{1}{\sqrt{\mu}}\frac{2^M}{2^N} \sum_{x} x_n\ketbra{\psi_x}. 
\]
Moreover, they yield $(N,M,p)$-QRACs with $p= \frac{1}{2}\left(1+\sqrt{\mu}\right)$.

\end{lemma}

\begin{proof}
By~(EC), $\Phi(\sigma^{(n)}) = \mu\, \sigma^{(n)}$; expanding $\Phi$ with~(UD),
\[
  \Phi(\sigma^{(n)})
    = \sqrt{\mu}\cdot\tfrac{2^M}{2^N}\sum_x {x_n}\,
  \ket{\psi_x}\!\bra{\psi_x},
\]
so $\frac{2^M}{2^N}\sum_x {x_n}\ket{\psi_x}\!\bra{\psi_x}
= \Phi(\sigma^{(n)})/\sqrt{\mu} = \sqrt{\mu}\,\sigma^{(n)}$. 

(\textup{UD}) implies that the worst-case decoding probability is given by
\begin{eqnarray}
    p=\min_{n,x}\tr{\ket{\psi_x}\!\bra{\psi_x}\frac{\idop+x_n\sigma^{(n)}}{2}}=\frac{1}{2}\left(1+\sqrt{\mu}\right).
\end{eqnarray}
\end{proof}

The family of QRACs in Ref.~\cite{suzuki2026analyticalconstructionnn1} also provides encoding states and decoding observables satisfying (TF), (UD) and (EC) at $\mu = M/(M+1)$. 
Moreover, we can show 
\begin{equation}
    \min_{x\in \{\pm\}^N}\topeigval{\sum_{n=1}^Nx_n\sigma^{(n)}}
    \leq \sqrt{\mu^*}:=\max_{x\in \{\pm\}^N}\topeigval{\sum_{n=1}^Nx_n\sigma^{(n)}}=\sqrt{\frac{M}{M+1}}.
\end{equation}
Thus, Eq.~\eqref{eq:worstprob} implies that the decoding probability of the \((M+1,M)\)-QRAC constructed in Ref.~\cite{suzuki2026analyticalconstructionnn1} is optimal for these fixed decoding observables, even when the encoding states are allowed to vary.


\subsection{New construction for $(M+2,M)$-QRAC}
While there exist \((M+1,M)\)-QRACs~\cite{suzuki2026analyticalconstructionnn1} attaining the conjectured bound in Eq.~\eqref{eq:num_conj}, the \((5,2)\)-QRACs discussed above suggest that analogous constructions are unlikely to hold for \((N,M)\)-QRACs with $N\geq M+3$.
However, thanks to MUPVMs, we can show the existence of $(M+2,M)$-QRACs achieving the conjectured bound, as below. 

\begin{theorem}
    For every $M \ge 1$, there exists an $(M+2, M, p)$-QRAC with 
    $$
    p=\frac{1}{2}\left(1 + \sqrt{\frac{M}{M+2}} \right).
    $$
\end{theorem}
\begin{proof}
    We prove this theorem by constructing an MUPVM family from observables $\{O_i\}_{i=1}^{M+2}$ and matching encoding states $\{\ket{\psi_b} \}_{b \in \{0,1\}^{M+2}}$ on $M$-qubit Hilbert space satisfying (TF), (UD), and (EC) conditions
    at $\mu = M/(M+2)$. 

    Let $\mathcal{K}$ be the $(M+2)$-qubit Hilbert space with computational basis $\{\ket{b}:b\in \{0,1\}^{M+2}\}$. Let $\Gamma, \Delta_i, \gamma_i, \rho_i$ be, resp., the diagonal parity operator, the coordinate-sign operators, the left, and the right Clifford generators all in $\mathcal{K}$ defined as 
    \begin{eqnarray*}
        \Gamma \ket{b} := (-1)^{\sum_i b_i}\ket{b}, &\qquad& \Delta_i \ket{b} := (-1)^{b_i} \ket{b}, \\
        \gamma_i \ket{b} := (-1)^{\sum_{j=1}^{i-1} b_j} \ket{b \oplus e_i}, &\qquad&
        \rho_i \ket{b} := (-1)^{\sum_{j=i+1}^{M+2}b_j}\ket{b \oplus e_i}.
    \end{eqnarray*}
    We can observe that they are Hermitian involutions with $\{\gamma_i, \gamma_j\} = 2\delta_{i,j}\idop$, $\left[\gamma_i, \rho_j\right] = 0$, $\Gamma \gamma_i = - \gamma_i \Gamma$, and $\Delta_i = \Gamma \gamma_i \rho_i$. 
    By using Pauli matrices, we can represent them as follows:
    \begin{equation}
        \Gamma=Z^{\otimes(M+2)},\ \Delta_i=Z_i,\ \gamma_i=Z^{\otimes(i-1)}\otimes X\otimes\idop^{\otimes(M+2-i)},\ 
        \rho_i=\idop^{\otimes(i-1)}\otimes X\otimes Z^{\otimes(M+2-i)}.
    \end{equation}

    Next, let us define orthonormal vectors $a, c \in \mathbb{R}^{M+2}$ such that $a_i = \sqrt{2/(M+2)} \cos\theta_i$ and $c_i = \sqrt{2/(M+2)}\sin\theta_i$ for $\theta_i = \pi(i-1)/(M+2)$, and construct the following orthogonal, traceless, Hermitian involutions 
    \[
    A := \sum_i a_i \gamma_i, \qquad 
    B := \sum_i c_i \gamma_i, \qquad
    C := \sqrt{-1} \Gamma B. 
    \]
    Note that they satisfy the following commutation relations.
    \begin{equation}
        \nonumber
        \{A,B\}=\{B,C\}=[A,C]=0.
    \end{equation}

    Then, the projection $P$ defined as 
    \[
    P := \frac{\idop+A}{2} \cdot \frac{\idop+C}{2}
    \]
    is a rank $2^{M}$ projection onto the simultaneous $+1$-eigenspace of $A$ and $C$. 
    We can verify that $P$ is a rank $2^M$ projection by observing $P^2=P$ and $\tr{P}=2^M$.
    Let any isometry $U: \mathbb{C}^{2^M} \to \Range(P)$ with $U^{\dagger}U = \idop_{2^M}$ and $UU^{\dagger} = P$. For $\lambda = \sqrt{\mu}$, the $M$-qubit observables of the MUPVMs and matching encoding (pure) quantum states are, respectively,  
    \[
    O_i := \lambda^{-1}U^{\dagger}P\Delta_i P U, \qquad \ket{\psi_b} := 2U^{\dagger}P\ket{b}. 
    \]

    By the identity $P = (\idop + A + C +AC)/4$ and $\sum_{j \neq i} \sin^2\left(\pi (i-j)/(M+2) \right) = (M+2)/2$, we can derive 
    \[
    \bra{b} P \ket{b} = \frac{1}{4}, \qquad \bra{b} P \Delta_i P \ket{b} = \frac{M}{4(M+2)}(-1)^{b_i}.
    \]
    Moreover, by decomposing $\gamma_i = a_i A + c_i B + \{\text{orthogonal terms}\}$ and the identity $P A P = P C P = P$, we obtain 
    \[
    P \Delta_i P \Delta_i P = \frac{M}{M+2} P.
    \]

    We are ready to verify the MUPVM property: each $O_i$ is clearly Hermitian ---because $\Delta_i$ and $P$ are---, an involution because 
    \[
    O_i^2 = \lambda^{-2} U^{\dagger}\left(P \Delta_i P \Delta_i P \right) U = \lambda^{-2} \cdot \frac{M}{M+2} U^{\dagger}P U = \idop_{2^M},
    \]
    and traceless because 
    \[
    \tr{O_i} = \lambda^{-1} \sum_b \bra{b}P\Delta_i P \ket{b} = \lambda^{-1}\frac{M}{4(M+2)}\sum_b (-1)^{b_i} = 0. 
    \]
    The identity $\tr{O_i O_j} = 0$ for distinct $i$ and $j$ can be verified as follows:
    \[
        \tr{O_iO_j}=\lambda^{-2}\tr{P\Delta_iP\Delta_j}=\frac{\lambda^{-2}M}{4(M+2)}\sum_b(-1)^{b_i+b_j}=0.
    \]
    Thus, $\{(\idop\pm O_i)/2 \}_{i=1}^{M+2}$ is a valid MUPVM family on $\mathbb{C}^{2^M}$. 

    The (TF) condition is satisfied because 
    \[
    \frac{2^M}{2^{M+2}} \sum_b \ket{\psi_b}\bra{\psi_b} = U^\dagger P \left(\sum_b \ket{b}\bra{b} \right) P U = U^\dagger P^2 U = U^\dagger P U = \idop_{2^M}. 
    \]

    The (UD) condition is from
    \[
    \bra{\psi_b} O_i \ket{\psi_b} = 4 \lambda^{-1} \bra{b} P \Delta_i P \ket{b} = \lambda^{-1} \frac{M}{M+2}(-1)^{b_i} = \sqrt{\mu} (-1)^{b_i}, 
    \]
    where we utilize $UU^\dagger = P$ and $P^3 = P$. 

    The (EC) condition is satisfied from
    \[
    \Phi(O_i) = \frac{2^M}{2^{M+2}} \sum_b \sqrt{\mu}(-1)^{b_i} \ket{\psi_b}\bra{\psi_b} = \sqrt{\mu} U^\dagger P \Delta_i P U = \mu O_i. 
    \]

    Thus, Lemma~\ref{lem:frame-povm} implies the existence of $(M+2, M, 1/2(1 + \sqrt{M/(M+2)}))$-QRACs. 
    
\end{proof}

\section{Discussion and open problems}

Our measurement-first approach shifts the study of QRACs from an information-theoretic analysis of encoding states to a spectral analysis of decoding measurements. 
In this formulation, the relevant quantities are the largest eigenvalues of the signed sums
\[
    Y_x=\sum_{n=1}^N x_n\sigma^{(n)} .
\]
We bounded these eigenvalues by modifying a Chernoff-type argument in two ways.

The first modification is to move from the worst-case setting to the average-case setting. 
The standard Chernoff-type argument directly controls tail probabilities and is therefore naturally suited to worst-case estimates. 
We observed that, when the moment-generating function \(f\) is convex, the same framework can also be used to bound the average-case decoding probability. 
The second modification is to enlarge the admissible class of moment-generating functions \(f\). 
By analyzing the Taylor coefficients of \(f\), we identified conditions under which the Chernoff-type argument remains valid. 
The resulting optimized functions have a simple form, making the new finite-size bounds explicit and numerically computable.

There remains room for further improvement. 
In the present argument, the matrix terms appearing in the Taylor expansion are bounded uniformly. 
A sharper analysis may be obtained by exploiting more of their noncommuting structure. 
For example, algebraic relations among the decoding observables could lead to tighter estimates on the norms or traces of the matrix terms. 
It would also be interesting to further investigate the connection between our refined Chernoff-type bounds and recent progress in matrix discrepancy theory. 
Such a connection may provide a route toward stronger finite-size QRAC bounds beyond those derived here.

The equality conditions of our QRAC bounds led naturally to the notion of mutually unbiased projector-valued measurements (MUPVMs). 
Thus, our proof provides not only an upper bound, but also a diagnostic for candidate optimal QRACs. 
At the same time, MUPVMs do not capture all constraints required for optimality. 
In the two-qubit case, attaining the refined bound imposes additional higher-order noncommuting relations among the decoding observables. 
These constraints go beyond pairwise mutual unbiasedness and involve the spectra of signed sums for all choices of signs. 
Understanding these higher-order relations may reveal deeper principles for designing optimal QRACs.

Indeed, several additional structures were exploited in our constructions.
First, in Propositions~\ref{prop:52} and~\ref{prop:53}, we used Pauli observables, or sparse sums of Pauli observables: anticommuting summands within each \(\sigma^{(n)}\) guarantee projectivity, while disjoint Pauli supports between different \(\sigma^{(n)}\)'s guarantee Hilbert--Schmidt orthogonality. 
This idea can also be used to obtain new families of QRACs as summarized in Table~\ref{tab:qrac-summary}. 
For example, we obtained \((7,2,0.7379)\)-, \((8,2,1/\sqrt{2})\)-, and \((9,2,0.6306)\)-QRACs whose success probabilities improve over those of Ref.~\cite{imamichi2018constructions}.
For the \((11,2)\)- and \((12,2)\)-QRACs, we observed that non-mutually unbiased PVMs outperform MUPVMs. 
This suggests that, when the conjectured bound is not attained, optimal QRACs might not be fully characterized by MUPVMs alone.
The details of these derivations will be presented in future work. 
Similar techniques can also be used to construct new QRACs with a small number of qubits. 
For instance, we found an \((8,5,p)\)-QRAC with 
\[ 
 \frac12\left(1+\sqrt{\frac35}\right) < p < \frac12\left(1 + \sqrt{\frac58} \right),
\]
which is strictly better than concatenating a \((5,3)\)-QRAC and a \((3,2)\)-QRAC but worse than the conjectured bound (Eq.~\eqref{eq:num_conj}).

Second, we used the TF, UD, and EC conditions, which are satisfied by Suzuki's QRAC construction~\cite{suzuki2026analyticalconstructionnn1}, to construct the \((M+2,M)\)-QRAC family. 
However, the construction of \((N,M)\)-QRACs attaining the conjectured bound (Eq.~\eqref{eq:num_conj}) for \(N\ge M+3\) remains unresolved. 
Our findings suggest that such a QRAC may not exist in this regime, as illustrated by the \((5,2)\)-QRAC case.

From a practical perspective, improved QRAC constructions may broaden the range of experimentally feasible quantum certification protocols and combinatorial optimization methods, as shown in Ref.~\cite{teramoto2023quantum}. 
Many such applications rely on a threshold quantum advantage, and better QRACs can reduce the resources needed to exceed these thresholds. 
From a theoretical perspective, the framework developed here---optimizing moment-generating functions for spectral bounds---may also be useful beyond QRACs. 
Examples include quantum state discrimination and random quantum states~\cite{M07,AHH12}, 
strong-converse and concentration arguments in quantum information theory~\cite{AW02}, 
dimension witnesses and prepare-and-measure certification~\cite{WCD08,THMB15,FK19,TSVBB20}, 
matrix discrepancy and vector balancing~\cite{DJR22,HRS22,BJM23}, 
Hamiltonian complexity and local Hamiltonian spectral estimates~\cite{KKR06,BH16}, 
many-body quantum chaos~\cite{MEHTA1960420,KLP18}, 
scrambling and black-hole information dynamics~\cite{SYK,PatrickHayden_2007}, 
and randomized measurement protocols~\cite{HKP20,EFHKPVZ23,HPS25}.

\section*{Acknowledgment}
SA is partially supported by JST Moonshot R\&D MILLENNIA Program (Grant No.JPMJMS2061), JPMXS0120319794, and CREST (Japan Science and Technology Agency) Grant No.JPMJCR2113.
ST is supported by JSPS KAKENHI grant JP22K11909. RR and SA acknowledge Takashi Imamichi of IBM Research Tokyo for his help in analyzing the $(5,2)$-QRAC. RR thanks Anupam Prakash for insightful discussions. 

\section*{AI Disclosure}
We used ChatGPT 5.5 Thinking Extended and Gemini 3.1 Pro to assist with the exploration and optimization of candidate QRAC constructions as well as the proof of nonexistence of an $(N,2)$-QRAC attaining the conjectured bound for $N>6$.
The final mathematical claims, proofs, numerical computations, and certificate tables were checked by the authors.
Some constructions are hand-verifiable, while others are certified by explicit computer-checkable tables.

\bibliographystyle{plain}
\bibliography{reference}

\appendix

\section{Summary of new QRAC constructions}
\label{appendix:QRAC_constructions}
\begin{table}[t]
\centering
\small
\begin{tabular}{r r c l}
\hline
QRAC & Success probability \(p\) & Decoder class & Verification \\
\hline
\((5,2)\)   & \(0.81182\)       & MUPVM & symbolic \\
\((7,2)\)   & \(0.73798\)       & MUPVM & symbolic \\
\((8,2)\)   & \(0.70711\)       & MUPVM & symbolic \\
\((9,2)\)   & \(0.63060\)       & MUPVM & symbolic \\
\((10,2)\)  & \(0.57265\)       & MUPVM & symbolic \\
\((11,2)\)  & \(\ge 0.53611\)   & PVM; not mutually unbiased & computer-certified \\
\((12,2)\)  & \(\ge 0.50707\)   & PVM; not mutually unbiased & computer-certified \\
\((7,3)\)   & \(\ge 0.81867\)   & MUPVM & computer-certified \\
\((8,3)\)   & \(\ge 0.79725\)   & MUPVM & computer-certified \\
\((10,3)\)  & \(\ge 0.73126\)   & MUPVM & computer-certified \\
\((11,3)\)  & \(0.71117\)       & MUPVM & symbolic \\
\((12,3)\)  & \(0.64434\)       & MUPVM & symbolic \\
\((7,4)\)   & \(\ge 0.87692\)   & MUPVM & computer-certified \\
\((9,4)\)   & \(\ge 0.82162\)   & MUPVM & computer-certified \\
\((8,5)\)   & \(\ge 0.88921\)   & MUPVM & computer-certified \\
\((9,5)\)   & \(\ge 0.87165\)   & MUPVM & computer-certified \\
\hline
\end{tabular}
\caption{Summary of the lower bounds of the success probabilities $p$ of the projective $(N,M)$-QRACs.}
\label{tab:qrac-summary}
\end{table}

We summarize in Table~\ref{tab:qrac-summary} several projective QRAC lower bounds obtained in this work.
Only the resulting
success probabilities and decoder classes are shown here; explicit
constructions and verification certificates will appear in a forthcoming
paper.  Here ``MUPVM'' means that the decoding observables can be
chosen as traceless Hermitian involutions satisfying
\(\operatorname{Tr}(\sigma^{(m)}\sigma^{(n)})=2^M\delta_{mn}\).
We provide a proof for an improved \((5,2)\)-QRAC whose success probability is better than that of Proposition~\ref{prop:52}.
\begin{proposition}
\label{prop:53}
    There exists a $(5,2,p)$-QRAC whose POVMs are MUPVMs and whose bias of success probabilities is 
\[
  \beta_{5,2}=2p-1
  =\sqrt{\frac{(\kappa+2)(\kappa^2+1)}{\kappa^3+4\kappa^2+3\kappa+6}} \approx 0.6236,
\]
where $\kappa\in(0,1)$ is the unique root of $\kappa^4+2\kappa^3-1=0$. 
\end{proposition}
\begin{proof}
    The decoding observables $\sigma^{(i)}$ are given by
    \begin{align}
  \sigma^{(1)}&=\cos\theta\, X \Id-\sin\theta\, Y Y,\label{eq:A1}\\
  \sigma^{(2)}&=\sin\theta\, X Y-\cos\theta\, Y \Id,\label{eq:A2}\\
  \sigma^{(3)}&=Z \Id,\label{eq:A3}\\
  \sigma^{(4)}&=\Id Z,\label{eq:A4}\\
  \sigma^{(5)}&=\Id X,\label{eq:A5}
\end{align}
where $\cos(2\theta) = \kappa^2$ for $\theta \in (0,\pi/4)$.
The quantum state for encoding $x \in \{\pm 1\}^5$ is defined as 
\begin{align}
    \ket{\psi_x} := \arg\max_{\ket{\psi}} \bra{\psi}~H_x~\ket{\psi},  
\end{align}
where 
$$
H_x := a x_1 \sigma^{(1)} + a x_2 \sigma^{(2)} +  b x_3 \sigma^{(3)} + c x_4 \sigma^{(4)} + c x_5 \sigma^{(5)}, 
$$
such that for $D := \kappa^3 + \kappa^2 + \kappa + 3$, 
$$
a := 1/D,~~b := 1 - 2(\kappa+1)/D,~~ c := \kappa/D.
$$
It can be confirmed that 
$$
\bra{\psi_x}~\sigma^{(i)}~\ket{\psi_x} = \beta x_i,
$$
holds for any $x \in \{\pm\}^5$ and $i = 1,\ldots,5$; hence $p_{\rm wc}=\frac{1}{2}(1+\beta)$. 
\end{proof}

For reference, some of the entries in Table~\ref{tab:qrac-summary}
come from the following closed-form biases $\beta_{N,M}$ for $(N,M)$-QRACs with success probabilities $p \ge \tfrac{1}{2}(1+\beta_{N,M})$:
\[
\begin{gathered}
\beta_{7,2}=\sqrt{(3-\sqrt2)/7},\qquad
\beta_{8,2}=\sqrt2-1,\qquad
\beta_{9,2}=\frac{2\sqrt2-1}{7},\\
\beta_{10,2}=\frac{1}{\sqrt{25+10\sqrt5}},\qquad
\beta_{11,3}=\frac{3\sqrt3+2\sqrt2}{19},\qquad
\beta_{12,3}=\frac1{2\sqrt3}.
\end{gathered}
\]
Further exact formulae and certificate details will be given in the forthcoming paper.

\end{document}